\documentclass[journal,final,10pt]{IEEEtran}

\usepackage{amsmath,graphicx,amsthm}
\usepackage{comment,amsmath,amssymb,epsfig,amsfonts,verbatim}
\usepackage{psfrag}
\usepackage{graphicx,times}
\usepackage{algorithmic}
\usepackage[boxed]{algorithm}
\usepackage{cite}
\usepackage{units}
\usepackage{color}
\usepackage{mathrsfs}
\usepackage{stfloats}
\usepackage{subfig}

\usepackage{tikz}
\usetikzlibrary{arrows,shapes}
\usetikzlibrary{positioning}

\usepackage{pgfplots}
\usepgfplotslibrary{fillbetween}
\pgfplotsset{compat=newest}
\newlength\figureheight
\newlength\figurewidth
\usepackage[final,markup=default,addedmarkup=none,deletedmarkup=sout]{changes}

\newtheorem{theorem}{Theorem}
\newtheorem{corollary}{Corollary}

\newtheorem{lemma}{Lemma}

\theoremstyle{definition}
\newtheorem{definition}{Definition}

\def\cc{c}
\def\rhob{\boldsymbol{\rho}}

\def\y{\mathbf y}

\def\n{{\mathbf n}}
\def\h{{\mathbf h}}
\def\g{{\mathbf g}}
\def\q{{\mathbf q}}

\def\H{{\bf H}}
\def\h{{\bf h}}
\def\I{{\bf I}}
\def\R{{\bf R}}
\def\Q{{\bf Q}}

\title{Improper Gaussian signaling for multiple-access channels in underlay cognitive radio}

\author{\IEEEauthorblockN{Christian Lameiro\authorrefmark{1},~\IEEEmembership{Member,~IEEE,} Ignacio Santamar\'ia\authorrefmark{2},~\IEEEmembership{Senior~Member,~IEEE,} and Peter J. Schreier\authorrefmark{1},~\IEEEmembership{Senior~Member,~IEEE}\\}
\thanks{\authorrefmark{1}C. Lameiro  and P. J. Schreier are with the Signal \& System Theory Group, Universit\"{a}t Paderborn, Germany (email: \{christian.lameiro, peter.schreier\}@sst.upb.de).}
\thanks{\authorrefmark{2}I. Santamar\'ia is with the Department of Communications Engineering,
University of Cantabria, Spain (e-mail: i.santamaria@unican.es).}}

\begin{document}
\maketitle
\begin{abstract}
 	This paper considers an unlicensed multiple-access channel (MAC) that coexists with a licensed point-to-point user, following the underlay cognitive radio paradigm. We assume that every transceiver except the secondary base station has one antenna, and that the primary user (PU) is protected by a minimum rate constraint. In contrast to the conventional assumption of proper Gaussian signaling, we allow the secondary users (SUs) to transmit improper Gaussian signals, which are correlated with their complex conjugate. When the secondary base station performs zero-forcing, we show that improper signaling is optimal if the sum of the interference channel gains (in an equivalent canonical model) is above a certain threshold. Additionally, we derive an efficient algorithm to compute the transmission parameters that attain the rate region boundary for this scenario. The proposed algorithm exploits a single-user representation of the secondary MAC along with new results on the optimality of improper signaling in the single-user case when the PU is corrupted by improper noise. 
\end{abstract}
\begin{keywords}
Improper Gaussian signaling, multiple-access channel, underlay cognitive radio.
\end{keywords}
\section{Introduction}
\label{sec:intro}
The number of devices with wireless connectivity has enormously increased in the last years, and the trend will continue in years to come. Cognitive radio (CR) has been proposed as an efficient means to satisfy the increasing demand on wireless resources. CR is built on the premise that the radio-frequency spectrum is underutilized. That is, the theoretical limits of the spectrum are far from being fully exploited. The CR paradigm consists in the primary users (PUs), which are the rightful owners of the spectrum, sharing their resources with unlicensed users, which are called secondary users (SUs), as long as the performance of the PUs is not compromised. Depending on how the PU is protected and the approach followed by the SUs to access the channel, there is interweave, overlay and underlay CR \cite{Goldsmith2009}. In this paper, we consider underlay CR, where the SUs are allowed to access the channel as long as the resulting interference at the primary receiver is below a tolerable threshold.

In wireless communications in general, and in underlay CR in particular, proper Gaussian signaling (PGS) is a common assumption in the study of theoretical limits and achievable rate regions. The reason behind this is that such signaling scheme is known to be optimal for the point-to-point, broadcast and multiple-access channels (BC and MAC, respectively) \cite{Cover}. In more sophisticated communication scenarios, such as the interference channel (IC), the optimal input distribution, except for some special cases, is still unknown. Indeed, when interference is treated as noise, the transmission of Gaussian signals that are improper rather than proper, i.e., correlated with their complex conjugate, has been shown to increase the performance in scenarios where interference is a limiting factor. As underlay CR is an interference-limited scenario (since the performance of the SUs is limited by the interference they cause to the PU), improper  Gaussian signaling (IGS) seems to be a promising approach. 

An improper complex random variable is, as opposed to a proper one, correlated with its complex conjugate \cite{Neeser1993}.\footnote{Impropriety is related to circularity. We say that a random variable $x$ is circular (or circularly symmetric) if its distribution is invariant under rotations of the form $e^{j\theta}x$, for any arbitrary constant $\theta$ \cite{Schreier}. While circularity implies propriety, the converse does not hold in general. As an exception, a Gaussian random variable that is proper is also circular.} Therefore, the differential entropy of an improper Gaussian random variable is lower than the differential entropy of its proper counterpart. Indeed, the maximum entropy theorem states that the entropy of a complex random variable, with a given covariance matrix, is maximized if the random variable is proper and Gaussian \cite{Neeser1993}, which explains the optimality of PGS in the aforementioned scenarios. However, when interference is treated as noise, the achievable rate decreases with increasing entropy of the interference. Thus, if the interference is improper, the achievable rate increases. Additionally, when the receiver is corrupted by improper noise, IGS becomes optimal for that user. Due to these facts, the usefulness of IGS in interference-limited scenarios becomes evident.

The payoffs of IGS over PGS were first shown by Cadambe \emph{et al.} \cite{Cadambe2010}. They studied the degrees-of-freedom (DoF), which characterize the asymptotic sum-capacity, of the single-antenna 3-user IC with constant channel extensions. They found that, while PGS provided 1 DoF, 1.2 DoF were achievable by IGS. A great deal of works has followed these lines showing the payoffs of IGS for different interference-limited networks, such as the IC \cite{Ho2012,Zeng2013,LagenMorancho2015} (also with discrete modulation schemes in \cite{Nguyen2015}), the Z-IC \cite{Lameiro2017,Lagen2014,Lagen2016,Kurniawan2015}, broadcast channel treating interference as noise \cite{Zeng2013_2,Park2013}, underlay and overlay CR \cite{LameiroSantamariaSchreier:2015:Benefits-of-Improper-Signaling-for-Underlay,LameiroSantamariaSchreier:2015:Analysis-of-maximally-improper-signalling,Amin2016,Amin2017,Gaafar2017}, relay channels \cite{Gaafar2016,Gaafar2016letter}, etc.

So far, the works analyzing IGS for underlay CR have focused on secondary networks comprised of a single SU. For example, in our previous work \cite{LameiroSantamariaSchreier:2015:Benefits-of-Improper-Signaling-for-Underlay} we consider a 2-user IC, where one of the users is the PU (which employs PGS) and the other one corresponds to the SU (which may use IGS). In this setting, and subject to a PU rate constraint, we derive the optimal transmission parameters of the SU and an insightful condition determining when IGS outperforms PGS. Similarly, \cite{Amin2016} considers the same scenario but assuming statistical channel state information (CSI), thus the analysis of this work focuses on the outage probability. A similar scenario is considered in \cite{Gaafar2017}, where a single SU coexists with a full-duplex PU. 

Due to the lack of results of IGS in multiuser secondary networks, this paper aims at providing insights in such scenarios. To this end, we consider a single PU sharing the spectrum with a secondary MAC (SMAC). The secondary users may use IGS, whereas the PU, unaware of any secondary users, employs PGS.  Our aim is to derive insights into the properties of IGS as well as its performance limits. Because of that, we assume that global CSI is available at the SUs. The results we obtain in this paper will then serve as design guidelines when more realistic assumptions on the availability of the CSI are considered. Additionally, our results can be applied outside the context of cognitive radio, when the point-to-point transmitter is restricted to PGS due to lack of CSI or because it is a legacy user (see, e.g., \cite{LagenMorancho2015}). For PGS, this scenario has been studied, e.g., in \cite{Zhang2009,Jorswieck2011,Ozcan2013,Sboui2015,Zhang2008cr}. A similar scenario, called partial interfering MAC (PIMAC), has been considered in \cite{Sezgin2017} with IGS and outside the context of cognitive radio. In that scenario, a point-to-point link coexists with a MAC, and the authors numerically optimized the IGS parameters when the MAC has two single-antenna users and a single-antenna base station. Here, we consider a similar setting as in \cite{Zhang2008cr}. That is, every user has one antenna, except for the secondary base station (BS), whose number of antennas is at least equal to the number of SUs. Unlike \cite{Zhang2008cr}, we consider a single PU but allow the SUs to employ IGS. \added{Our main contributions are summarized next.}

\begin{itemize}
 \item \added{In contrast to existing works, we consider IGS for a multiuser rather than a single-user secondary network, the SMAC. Assuming that the interference must be such that the PU achieves a certain minimum rate, we provide a complete characterization of the rate-region boundary of the SUs when the BS employs a zero-forcing (ZF) decoding scheme. For this setting, we show that the system behaves as in an equivalent single-SU scenario. We leverage this observation to derive a necessary and sufficient condition for a given boundary point to be achieved by IGS. This condition is stated as a threshold on the sum of the interference channel coefficients (in an equivalent canonical model), above which IGS outperforms PGS.}
 
 \item \added{We provide an algorithm to compute the transmission parameters that achieve the boundary of the rate region. The algorithm is based on closed-form expressions rather than numerical optimization, hence permitting insights into the behavior of IGS. The proposed algorithm successively identifies the SUs for which either the constraint on power budget or on the degree of impropriety is active. This subset of users has known parameters and the optimization has to be carried out only for the remaining users. The key idea is again the transformation of the SMAC into an equivalent single-user channel. This makes it possible to apply our previous results for the single-user case \cite{LameiroSantamariaSchreier:2015:Benefits-of-Improper-Signaling-for-Underlay}. However, in this equivalent single-user representation of the SMAC, the equivalent noise at the primary receiver may be improper, which requires extending our results in \cite{LameiroSantamariaSchreier:2015:Benefits-of-Improper-Signaling-for-Underlay} to this case. Therefore, as a byproduct, our paper also extends the results in \cite{LameiroSantamariaSchreier:2015:Benefits-of-Improper-Signaling-for-Underlay} to the case where the PU is corrupted by improper noise, which is a relevant result on its own and hence represents an additional contribution.}
\end{itemize}

The rest of the paper is organized as follows. Section \ref{sec:model} introduces the system model. The optimal signaling for a single SU coexisting with a PU affected by improper noise is analyzed in Section \ref{sec:ii}. Section \ref{sec:opt} formulates the characterization of the rate region boundary under ZF decoding and, using the results obtained in Section \ref{sec:ii}, derives the optimal transmission parameters attaining this boundary. Numerical examples are provided in Section \ref{sec:sim}. 

\section{System model}\label{sec:model}
\subsection{Preliminaries about improper random variables}
We start by providing the necessary background on improper random variables. We refer the reader to \cite{Schreier} for a comprehensive treatment of the topic.
\begin{definition}[\hspace{-0.25pt}\cite{Schreier}]
	The complementary variance of a zero-mean complex random variable $x$ is defined as $\tilde{\sigma}_x=\operatorname{E}[x^2]$. If $\tilde{\sigma}_x=0$, then $x$ is called proper, otherwise improper. 
\end{definition}
Furthermore, $\sigma_x^2$ and $\tilde{\sigma}_x$ are a valid pair of variance and complementary variance if and only if $\sigma_x^2\geq0$ and $|\tilde{\sigma}_x|\leq\sigma_x^2$.
\begin{definition}[\hspace{-0.25pt}\cite{Schreier}]\label{def:cc}
	The circularity coefficient of a complex random variable $x$, which measures the degree of impropriety, is defined as $\cc_x=|\tilde{\sigma}_x|/\sigma_x^2$. The circularity coefficient satisfies $0\leq\cc_x\leq1$. If $\cc_x=0$, then $x$ is proper, otherwise improper. If $\cc_x=1$ we call $x$ maximally improper.
\end{definition}

\subsection{System description}
\begin{figure}[t!]
\centering
\subfloat[]{\includegraphics[width=1\columnwidth]{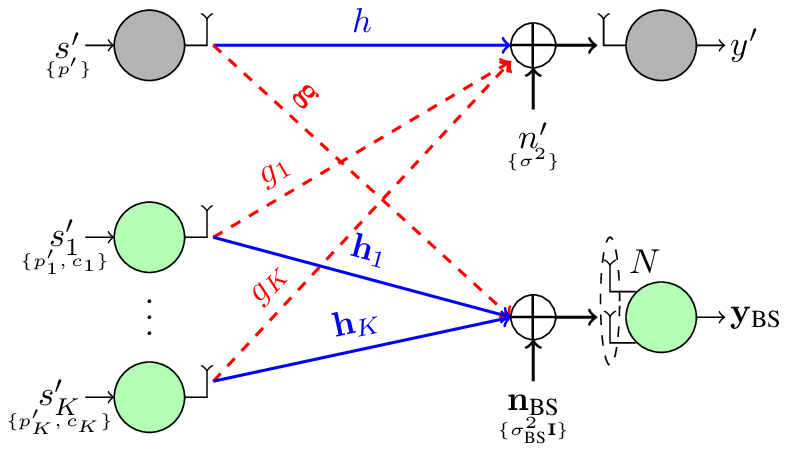}\label{fig:SMACi}}\\%
\subfloat[]{\includegraphics[width=1\columnwidth]{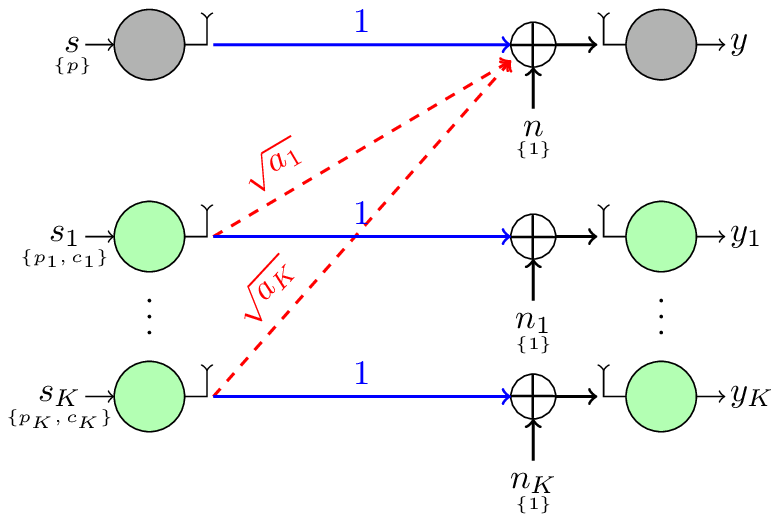}\label{fig:SMACcm}}%
\caption{(a) Considered underlay cognitive radio scenario and (b) its canonical model after ZF-SIC. The variance (or covariance matrix) and the circularity coefficient of the signals and noise are given in brackets. The top link (gray nodes) is the PU, and the system at the bottom (green nodes) is the SMAC. Additionally, the nodes on the left-hand side are the transmitters, and those on the right-hand side are the receivers.}%
\label{fig:SMAC}%
\end{figure}
\added{We consider a primary point-to-point link transmitting at a fixed rate $\bar{R}$, where both transmitter and receiver have a single antenna. Located in the vicinity of the PU, there is an SMAC comprised of $K$ SUs and a secondary base station (BS), as depicted in Fig. \ref{fig:SMACi}. The SUs access the same frequency band used by the PU, hence interference between the two systems is created. Following the underlay CR principle, the SUs are allowed to access the channel as long as the PU performance is not compromised, i.e., provided the generated interference is such that the PU can support the prescribed data rate $\bar{R}$.} We consider quasi-static fading channels, i.e., the channels remain constant during the transmission of a codeword. The signal at the primary receiver can be expressed as
\begin{equation}\label{eq:y}
	y'=\sqrt{p'}hs'+\sum_{k=1}^K\sqrt{p'_k}g_k s'_k+n' \; ,
\end{equation}
where $p'$, $h$, and $s'$ are the transmit power, direct channel, and transmit symbol of the PU; $p'_k$, $g_k$, and $s'_k$ denote the transmit power, cross-channel, and transmit symbol of the $k$th SU; and $n'$ is additive white Gaussian noise (AWGN) with variance $\sigma^2$. For all the relevant parameters, we indicate with the subscript $k$ the $k$th SU, whereas no subscript is used for the PU. We assume that the noise and the PU transmit signal are proper complex Gaussian random variables. The SUs, however, transmit complex Gaussian signals that are allowed to be improper with circularity coefficient $c_k$, $k=1,\ldots,K$. Notice that, even though the SUs may benefit if the PU also uses IGS, the primary transmitter is unaware of the secondary system, and thus sticks to PGS, which is a common assumption \cite{LameiroSantamariaSchreier:2015:Benefits-of-Improper-Signaling-for-Underlay,LameiroSantamariaSchreier:2015:Analysis-of-maximally-improper-signalling,Amin2016,Amin2017,Gaafar2017}. Furthermore, the primary transmitter does not require any side information (such as channel state information) if it follows a PGS scheme. This would not be the case if the PU used IGS.

The signal at the secondary BS can be expressed as 
\begin{equation}
	\y_{\text{BS}}=\sum_{k=1}^K\sqrt{p'_k}\h_ks'_k+\sqrt{p'}\g s'+\n_{\text{BS}} \; ,
\end{equation}
where $\n_{\text{BS}}$ is AWGN, which is assumed to be proper with covariance matrix $\sigma_{\text{BS}}^2\I$, and $\h_k$ and $\g$ are the direct channel of the $k$th SU and the PU cross-channel, respectively. The BS applies ZF successive interference cancellation (ZF-SIC), so that the interference between SUs is removed. Even though the minimum mean square error (MMSE) SIC receiver achieves the capacity of the MAC \cite{Varanasi1997}, a ZF-based decoder is sometimes preferable (see \cite{Zhang2008cr}), as it allows a more tractable analysis since the SUs are only coupled through the PU rate constraint. This permits drawing insights into the superiority of IGS that otherwise would go unnoticed. Additionally, it has been shown in \cite{Caire2003} that ZF-SIC is sum-capacity achieving in the asymptotic high and low signal-to-noise ratio (SNR) regimes.\footnote{The asymptotic optimality is shown in \cite{Caire2003} for the ZF dirty paper coding scheme in the broadcast channel. This result can be translated into the MAC with ZF-SIC as they are duals of each other \cite{Tran2013}.} The ZF-SIC decoder works as follows \cite{Zhang2008cr,Tran2013}.\footnote{We assume that the secondary BS is unaware of the primary transmitter and hence has no knowledge of $\g$, so the interference from the PU is treated as extra proper Gaussian noise. Nevertheless, our results are also valid for any other ZF decoding scheme.} First, the QR decomposition of the channel matrix $\H=[\h_1,\ldots,\h_K]$ is computed, which yields the unitary matrix $\Q$ and the upper-triangular matrix $\R$. The received signal is then multiplied by the unitary decoder $\Q^H$, and SIC is further carried out. Notice that, in our problem, optimizing over the user ordering results in a combinatorial problem. This is because the user ordering affects the equivalent channels (given by the diagonal elements of $\R$), as well as the interference from the PU (given by the projection of $\g$ onto each column of $\Q$), and hence all the $K!$ user orderings should be evaluated. Since this is of prohibitive complexity, we consider a fixed user ordering given without loss of generality by the reversed user indexes, as also done in other works (see, e.g., \cite{Zhang2008cr} and \cite{Caire2003}). Thus, the equalized signals are
\begin{equation}\label{eq:yk}
	y'_k=\sqrt{p'}_kr_ks'_k+\sqrt{p'}\q_k^H\g s'+\q_k^H\n_{\text{BS}} \; , \; k=1,\ldots,K \; ,
\end{equation}
where $r_k$ is the $k$th diagonal entry of $\R$, and $\q_k$ is the $k$th column of $\Q$. For compactness of the expressions and ease of their interpretation, it is sometimes useful to use the canonical model, where the noise variances and direct channel gains are set to one, as we depict in Fig. \ref{fig:SMACcm}. The canonical model is obtained as follows. Let us define $s=e^{j\angle{h}}s'$ and $s_k=e^{j\angle{g_k}}s'_k$, where $\angle{\cdot}$ denotes the phase of a complex scalar. Since $s$ is proper Gaussian, its distribution is exactly equal to that of $s'$. The distribution of $s_k$ is equal to that of $s'_k$ except for the phase of the complementary variance, which is $\phi_k=2\angle{g_k}+\phi'_k$, where $\phi'_k$ is the phase of the complementary variance of $s'_k$. Additionally, we define $n'_k=\sqrt{p'}\q_k^H\g s'+\q_k^H\n_{\text{BS}}$, whose variance is $\sigma_k^2=p'|\q_k^H\g|^2+\sigma_{\text{BS}}^2$. Now we divide \eqref{eq:yk} by $\sigma_k$ and multiply it by $e^{j(\angle{g_k}-\angle{r_k})}$ , obtaining
\begin{align}\label{eq:can1}
	e^{j(\angle{g_k}-\angle{r_k})}\frac{y'_k}{\sigma_k}=y_k&=\sqrt{\frac{p'_k|r_k|^2}{\sigma_k^2}}e^{j\angle{g_k}}s'_k+e^{j(\angle{g_k}-\angle{r_k})}\frac{n'_k}{\sigma_k}\notag\\&=\sqrt{p_k}s_k+n_k \; , \; k=1,\ldots,K \; .
\end{align}
Notice that $p_k=\frac{p'_k|r_k|^2}{\sigma_k^2}$ and that $n_k$ is proper with unit variance. Similarly, we divide \eqref{eq:y} by the noise standard deviation $\sigma$ to obtain
\begin{align}\label{eq:can2}
	\frac{y'}{\sigma}=y&=\sqrt{\frac{p'|h|^2}{\sigma^2}}e^{j\angle{h}}s'+\sum_{k=1}^K\sqrt{\frac{p'_k|g_k|^2}{\sigma^2}}e^{j\angle{g_k}} s'_k+\frac{n'}{\sigma}\notag\\&=\sqrt{p}s+\sum_{k=1}^K\sqrt{\frac{p_k\sigma_k^2|g_k|^2}{\sigma^2|r_k|^2}} s_k+n \; .
\end{align}
Notice again that $p=\frac{p'|h|^2}{\sigma^2}$ and that $n$ is proper with variance one. Finally, we define $a_k=\frac{\sigma_k^2|g_k|^2}{\sigma^2|r_k|^2}$ and obtain the canonical model from \eqref{eq:can1} and \eqref{eq:can2} as
\begin{align}
	y&=\sqrt{p}s+\sum_{k=1}^K\sqrt{p_k}\sqrt{a_k}s_k+n \; ,\\
	y_k&=\sqrt{p_k}s_k+n_k \; , \; k=1,\ldots,K \; .	
\end{align}

\added{Since the PU uses PGS, but is affected by improper interference, its achievable rate can be expressed as (see \cite[Eq. (29)]{Zeng2013})}
\begin{align}
	&R(\{p_k,\cc_k,\phi_k\}_{k=1}^K)=\notag\\&\frac{1}{2}\log_2\left[\frac{\left(1+p+\sum_{k=1}^Ka_kp_k\right)^2-\left|\sum_{k=1}^Ka_kp_k\cc_ke^{j\phi_k}\right|^2}{\left(1+\sum_{k=1}^Ka_kp_k\right)^2-\left|\sum_{k=1}^Ka_kp_k\cc_ke^{j\phi_k}\right|^2}\right] \; ,\label{eq:Rpu}
\end{align}
\added{where the complementary variance of the $k$th SU signal $s_k$ is, by Definition \ref{def:cc}, expressed as $\tilde{p}_k=p_k\cc_ke^{j\phi_k}$,} with $c_k$ and $\phi_k$ being the circularity coefficient and phase of the complementary variance, respectively. Similarly, \added{as the SUs use IGS but are affected by proper noise,} the rate achieved by the $k$th SU can be written as (see  \cite[Eq. (7)]{LameiroSantamariaSchreier:2015:Benefits-of-Improper-Signaling-for-Underlay})
\begin{equation}
	R_k(p_k,\cc_k,\phi_k)=\frac{1}{2}\log_2\left\{1+p_k\left[p_k\left(1-\cc_k^2\right)+2\right]\right\} \; .\label{eq:Rk0}
\end{equation}

Notice that $R_k(p_k,\cc_k,\phi_k)=R_k(p_k,\cc_k)$, i.e., the SU rate is not a function of $\phi_k$. Therefore, $\{\phi_k\}_{k\in\mathcal{K}}$ must be chosen to maximize the PU rate. We notice that
\begin{align}
	&\left(1+p+\sum_{k\in\mathcal{K}}a_kp_k\right)^2>\left(1+\sum_{k\in\mathcal{K}}a_kp_k\right)^2 \notag\\&\Rightarrow \; \frac{\partial R(\{p_k,\cc_k,\phi_k\}_{k\in\mathcal{K}})}{\partial \left|\sum_{k\in\mathcal{K}}a_kp_k\cc_ke^{j\phi_k}\right|^2}>0 \; ,
\end{align}
which is always true. Therefore, maximizing $R$ in $\phi_k$ is equivalent to maximizing \linebreak $\left|\sum_{k\in\mathcal{K}}a_kp_k\cc_ke^{j\phi_k}\right|$. Due to the triangle inequality, we have $\left|\sum_{k\in\mathcal{K}}a_kp_k\cc_ke^{j\phi_k}\right|\leq\sum_{k\in\mathcal{K}}a_kp_k\cc_k$, with equality for
\begin{equation}\label{eq:phiopt}
	\phi_k^\star=\phi \; , \; k=1,\ldots,K \; ,
\end{equation}
i.e., the phase of the complementary variances of the SUs are equal to an arbitrary value $\phi$. Therefore, the optimization needs to be carried out in the transmit powers $p_k$ and circularity coefficients $\cc_k$, $k=1,\ldots,K$. Eq. \eqref{eq:phiopt} can be interpreted by looking at the joint distribution of the real and imaginary parts of each SU signal at the primary receiver. Their probability density contours are ellipses whose major axes are rotated by $\phi_k/2$ \cite{Schreier}. Hence, Eq. \eqref{eq:phiopt} aligns all interference signals along the same dimension, which in turn maximizes the circularity coefficient of the aggregate interference.

In the next sections we characterize the rate region boundary of this SMAC subject to a PU rate constraint. The adopted approach can be summarized as follows.
\begin{enumerate}
	\item The SMAC is first transformed into an equivalent single-user channel, which yields the model in \cite{LameiroSantamariaSchreier:2015:Benefits-of-Improper-Signaling-for-Underlay}. This permits applying those results to derive the optimality condition for IGS.
	\item The optimal transmission parameters, i.e., $p_k$ and $\cc_k$ $\forall k$, are found in closed form by first identifying, one by one, the users for which either the power budget constraint or the constraint on the circularity coefficient is active. We show how the parameters of these users can easily be obtained.
	\item After identifying one of the aforementioned users, the equivalent single-user representation has to be modified. Specifically, the noise at the primary receiver in the equivalent single-user model becomes improper, and thus the equivalent system no longer follows the model in \cite{LameiroSantamariaSchreier:2015:Benefits-of-Improper-Signaling-for-Underlay} (where the noise is proper).
	\item In order to analyze the updated single-user representation, we have to extend the results in \cite{LameiroSantamariaSchreier:2015:Benefits-of-Improper-Signaling-for-Underlay} to the case where the PU is affected by improper noise. This is carried out in Section \ref{sec:ii}, before the SMAC is analyzed (Section \ref{sec:opt}), so as not to break the flow of the analysis.
\end{enumerate}


\section{Single-user case with a primary user affected by improper noise}\label{sec:ii}
\begin{figure}[t!]
\centering
\includegraphics[width=1\columnwidth]{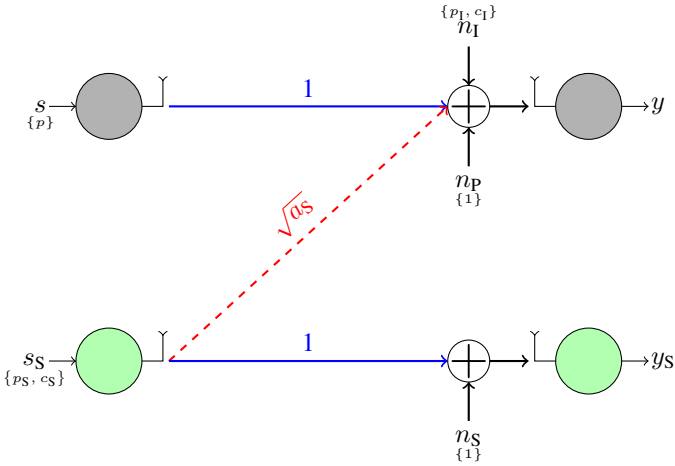}
\caption{Single-SU scenario with a PU affected by improper noise.}
\label{fig:PUwII}
\end{figure}
Let us consider a single SU and assume that the primary receiver is affected by improper noise, as depicted in Fig. \ref{fig:PUwII}. This model will be obtained as the single-user representation of the SMAC, which will be derived in Section \ref{sec:opt}. Therefore, the analysis of this scenario is key to characterize the rate-region boundary of the SMAC. Notice also that this model generalizes the one considered in \cite{LameiroSantamariaSchreier:2015:Benefits-of-Improper-Signaling-for-Underlay}, where the noise at the PU is assumed to be proper. In order to clearly differentiate this scenario from the one having an SMAC, we indicate the parameters of this user with the subscript ``S''. For convenience, we express the improper noise at the PU as a summation of two terms, namely $n=n_{\text{P}}+n_{\text{I}}$, where $n_{\text{P}}$ is proper Gaussian with variance 1, and $n_{\text{I}}$ is improper Gaussian with variance $p_{\text{I}}$ and circularity coefficient $\cc_{\text{I}}>0$, which is independent of $n_{\text{P}}$. The reason behind this decomposition of the noise is that it matches the structure of the noise in the equivalent single-user representation of the SMAC, which we will derive in the next section. Notice also that, by \eqref{eq:phiopt}, the optimal phase $\phi_{\text{S}}^\star$ of the complementary variance of the SU equals $\phi_{\text{I}}$, with $\phi_{\text{I}}$ being the phase of the complementary variance of $n_{\text{I}}$. Therefore, the parameters become again the transmit power and the circularity coefficient. 

Our goal now is to obtain $p_{\text{S}}$ and $\cc_{\text{S}}$ to maximize the SU rate while ensuring the PU rate constraint. That is, we want to solve the problem
\begin{align}
\mathcal{P}_{\text{IN}}:\hspace{0.5cm} & \underset{p_{\text{S}},\cc_{\text{S}}}{\text{maximize}}
& & R_{\text{S}}\left(p_{\text{S}},\cc_{\text{S}}\right) \; ,\notag\\
& \text{subject to}
& & 0\leq p_{\text{S}} \leq P_{\text{S}} \; , \notag\\
& & & 0\leq\cc_{\text{S}}\leq1 \; , \notag\\
& & & R\left(p_{\text{S}},\cc_{\text{S}}\right)\geq\bar{R} \; , \notag
\end{align}
where IN stands for improper noise. $R_{\text{S}}(p_{\text{S}},\cc_{\text{S}})$ follows \eqref{eq:Rk0} and the PU rate is now
\begin{equation}\resizebox{0.49\textwidth}{!}{$
	R(p_{\text{S}},\cc_{\text{S}})=\frac{1}{2}\log_2\left[\frac{\left(1+p+a_{\text{S}}p_{\text{S}}+p_{\text{I}}\right)^2-\left(a_{\text{S}}p_{\text{S}}\cc_{\text{S}}+p_{\text{I}}\cc_{\text{I}}\right)^2}{\left(1+a_{\text{S}}p_{\text{S}}+p_{\text{I}}\right)^2-\left(a_{\text{S}}p_{\text{S}}\cc_{\text{S}}+p_{\text{I}}\cc_{\text{I}}\right)^2}\right] \; .$}\label{eq:RpuII}
\end{equation}
The first constraint in $\mathcal{P}_{\text{IN}}$ is the power budget constraint, whereas the second ensures that the circularity coefficient is valid. The last constraint is the PU rate constraint, which guarantees that the PU performance is not compromised. Notice that $\bar{R}$ has to be small enough so that $\mathcal{P}_{\text{IN}}$ is feasible. Specifically, $\mathcal{P}_{\text{IN}}$ is feasible as long as $\bar{R}\leq R(0,0)$, i.e., the required rate cannot be greater than the achievable rate in the absence of interference. We now use \eqref{eq:RpuII} to express the rate constraint, which is the last constraint in $\mathcal{P}_{\text{IN}}$, after some manipulations as
\begin{align}\label{eq:psqs}
	&a_{\text{S}}^2p_{\text{S}}^2(1-\cc_{\text{S}}^2)\leq(1-\beta)[p+2(1+a_{\text{S}}p_{\text{S}}+p_{\text{I}})]-1\notag\\&-2(a_{\text{S}}p_{\text{S}}+p_{\text{I}})-p_{\text{I}}^2(1-\cc_{\text{I}}^2)-2a_{\text{S}}p_{\text{S}}p_{\text{I}}(1-\cc_{\text{S}}\cc_{\text{I}}) \; ,
\end{align}
where 
\begin{equation}
	\beta=1-\frac{p}{2^{2\bar{R}}-1} \; .\label{eq:beta}
\end{equation}
Let $q(\cc_{\text{S}})$ be the value of $p_{\text{S}}$ that makes \eqref{eq:psqs} hold with equality, which is the maximum transmit power of the SU such that the PU rate constraint is satisfied. The PU rate constraint can then be replaced by $p_{\text{S}}\leq q(\cc_{\text{S}})$, and the SU maximizes its rate by taking $p_{\text{S}}=\min[q(\cc_{\text{S}}),P_{\text{S}}]$, where $P_{\text{S}}$ is the SU power budget. Let us first analyze the behavior of $R_{\text{S}}(p_{\text{S}}=q(\cc_{\text{S}}),\cc_{\text{S}})=R_{\text{S}}(\cc_{\text{S}})$. To this end, we present the following lemma. 
\begin{lemma}\label{th:l1}
	There exists $0<\cc_{\text{\normalfont R}}\leq 1$ such that $R_{\text{\normalfont S}}(\cc_{\text{\normalfont S}})$ is monotonically increasing for $0\leq\cc_{\text{\normalfont S}}\leq\cc_{\text{\normalfont R}}$ and monotonically decreasing for $\cc_{\text{\normalfont R}}<\cc_{\text{\normalfont S}}\leq 1$.
\end{lemma}
\begin{proof}	
	Please refer to Appendix \ref{app:l1}.	
\end{proof}

Lemma \ref{th:l1} implies a surprising result. If the power budget is sufficiently high, the SU improves its rate by using IGS independently of the interference channel coefficient $a_{\text{S}}$ (as long as $a_{\text{S}}>0$ and $\cc_{\text{I}}>0$). This is the case as long as PGS does not permit maximum power transmission, in which case the transmit power can increase by making the transmit signal improper. Nevertheless, this is still rather surprising and in sharp contrast with the results obtained under proper noise in \cite{LameiroSantamariaSchreier:2015:Benefits-of-Improper-Signaling-for-Underlay,Lameiro2017}, where IGS is shown to improve the rate only when $a_{\text{S}}$ is above a given threshold. As a matter of fact, it is well-known that IGS is optimal if the own receiver is affected by improper noise. This new result establishes that IGS is also optimal (in the sense of Pareto) if the receiver that is interfered with is affected by improper noise (even though the transmitter associated with that receiver may be using PGS). Lemma \ref{th:l1} also states that there are two possible behaviors for $R_{\text{S}}(\cc_{\text{S}})$. Either $\cc_{\text{R}}=1$, so that the rate increases with $\cc_{\text{S}}$, in which case the maximum is achieved at $\cc_{\text{S}}=1$; or $0<\cc_{\text{R}}<1$ and then the rate increases in the interval $0\leq\cc_{\text{S}}\leq\cc_{\text{R}}$, in which case the maximum is achieved at $0<\cc_{\text{S}}=\cc_{\text{R}}<1$. When the power budget is taken into account, the optimal circularity coefficient will not be greater than the value of $\cc_{\text{S}}$ for which $q(\cc_{\text{S}})=P_{\text{S}}$. \added{With these observations, we state the optimal solution of $\mathcal{P}_{\text{IN}}$ in the following theorem.}

\begin{theorem}\label{th:th1}
	The optimal circularity coefficient, which is the optimal solution of $\mathcal{P}_{\text{\normalfont IN}}$, is given by
	\begin{equation}\label{eq:kStar}
		\cc_{\text{\normalfont S}}^\star=\min(\cc_{\text{\normalfont B}},\cc_{\text{\normalfont R}}) \; ,
	\end{equation}
	where 
	\begin{equation}\label{eq:kP}
		\cc_{\text{\normalfont B}}=\max\{0\leq\cc_{\text{\normalfont S}}\leq1: q(\cc_{\text{\normalfont S}})\leq \max[P_{\text{\normalfont S}},q(0)]\} \; ,
	\end{equation}
	and
	\begin{align}\label{eq:kR}
		\cc_{\text{\normalfont R}}=
		\max\Bigg\{0\leq\cc_{\text{\normalfont S}}\leq1:& a_{\text{\normalfont S}}q(\cc_{\text{\normalfont S}})\cc_{\text{\normalfont S}}\left(1-\frac{p_{\text{\normalfont I}}+\beta}{a_{\text{\normalfont S}}}\right)\notag\\&+p_{\text{\normalfont I}}\cc_{\text{\normalfont I}}\left[1+q(\cc_{\text{\normalfont S}})\right]\geq0\Bigg\} \; ,
	\end{align}	
	where $\beta$ is given by \eqref{eq:beta}.
\end{theorem}
\begin{proof}
	Please refer to Appendix \ref{app:th1}.	
\end{proof}
\begin{corollary}\label{th:th11}
	$\cc_{\text{\normalfont B}}$, $\cc_{\text{\normalfont R}}$, and $q(\cc_{\text{\normalfont S}})$ admit a closed-form expression. Additionally, $\cc_{\text{\normalfont R}}=1$ if and only if $a_{\text{\normalfont S}}\geq\xi$, where
	\begin{equation}\label{eq:mu}		
		\xi=\frac{\left[p_{\text{\normalfont I}}(1-\cc_{\text{\normalfont I}})+\beta\right]\left\{\bar{p}^2-\left[p_{\text{\normalfont I}}(1-\cc_{\text{\normalfont I}})+\beta\right]\left[p_{\text{\normalfont I}}(1+\cc_{\text{\normalfont I}})+\beta\right]\right\}}{\bar{p}^2-\left[p_{\text{\normalfont I}}(1-\cc_{\text{\normalfont I}})+\beta\right]^2} \; ,
	\end{equation}
	with $\bar{p}=\frac{p2^{\bar{R}}}{2^{2\bar{R}}-1}$.
\end{corollary}
\begin{proof}
	Please refer to Appendix \ref{app:th11}.	
\end{proof}

Theorem \ref{th:th1} provides the optimal circularity coefficient of the SU, and it can be clearly seen that, if $q(0)<P_{\text{S}}$, IGS is optimal independently of any other system parameter. This is because, in such a case, $\cc_{\text{B}}>0$ and $\cc_{\text{R}}>0$. The latter is due to the fact that $p_{\text{I}}>0$ and $\cc_{\text{I}}>0$, as can clearly be seen in \eqref{eq:kR}. \added{Additionally, the optimal transmit power can be obtained as $p_{\text{S}}^\star=q(\cc_{\text{S}}^\star)$.} 

Fig. \ref{fig:ExampleRsu} illustrates the behavior of the rate with $\cc_{\text{S}}$ and for different values of the interference channel coefficient $a_{\text{S}}$. As stated in Lemma \ref{th:l1}, $R_{\text{S}}(\cc_{\text{S}})$ is increasing in the interval $0\leq\cc_{\text{S}}\leq\cc_{\text{R}}$, with $\cc_{\text{R}}$ indicated with a circle in Fig. \ref{fig:ExampleRsu}. The parameters for this example are $p=100$, $p_{\text{I}}=5$, $\cc_{\text{I}}=0.5$ and $\bar{R}=3.31\,\text{b/s/Hz}$, which makes $\xi=2.16$. 
\begin{figure}[t!]
\centering
\includegraphics[width=1\columnwidth]{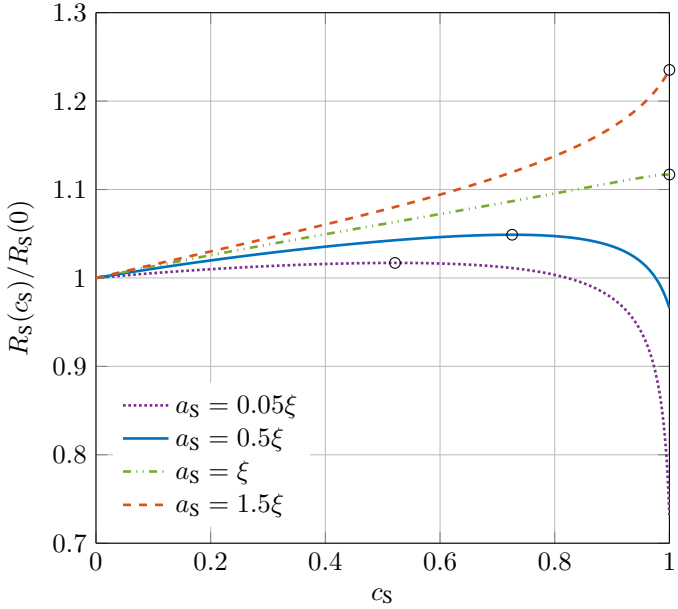}
\caption{Example of the SU rate normalized by the rate achieved by PGS for different values of $a_{\text{S}}$. We indicate the point $\cc_{\text{S}}=\cc_{\text{R}}$ with a circle.}
\label{fig:ExampleRsu}
\end{figure}

In the next section we will characterize the rate region boundary of the considered SMAC. This will be accomplished by deriving an equivalent single-user representation to which we apply the results of this section.

\section{Rate region boundary}\label{sec:opt}
We return to the general case with $K$ SUs forming an SMAC, as detailed in Section \ref{sec:model} (see also Fig. \ref{fig:SMAC}). In order to compute the boundary of the rate region, we use the rate profile approach \cite{Mohseni2006}. Following these lines, a given point of the rate region, characterized by a given set of non-negative quantities $\{\alpha_k\}_{k=1}^K$, with $\sum_{k=1}^K\alpha_k=1$, can be computed by solving the following optimization problem.
\begin{align}
\mathcal{P}:\hspace{0.5cm} & \underset{\{p_k,\cc_k\}_{k=1}^K,r}{\text{maximize}}
& & r \; ,\notag\\
& \text{subject to}
& & 0\leq p_k \leq P_k \; , \; k=1,\ldots,K \; , \notag\\
& & & 0\leq\cc_k\leq1 \; , \; k=1,\ldots,K \; ,\notag\\
& & & R_k(p_k,\cc_k)\geq\alpha_kr \; , \; k=1,\ldots,K \; ,\notag \\
& & & R(\{p_k,\cc_k\}_{k=1}^K)\geq\bar{R} \; , \notag
\end{align}
where $\bar{R}$ is the PU rate constraint (with $\bar{R}\leq R(\{0,0\}_{k=1}^K)$ so that the problem is feasible) and we have taken the optimal phases $\phi_k=\phi$ $\forall k$ (see \eqref{eq:phiopt}). Our aim is to solve the above problem analytically, so that insightful conclusions about the optimality of IGS can be drawn. Note also that these insights can also be transferred to the case where the secondary BS performs MMSE-SIC. This is because with the MMSE-SIC scheme there is also interference between the SUs, and hence IGS will also mitigate the impact of this interference. Therefore, if IGS outperforms PGS for ZF-SIC and a given set of canonical channel coefficients $a_1,\ldots,a_K$, it will also do so for MMSE-SIC and the same set of channel coefficients. 

If $\alpha_k=0$ for some $k$, the optimal transmit power for that user is zero, i.e., $p_k=0$. Therefore, we define $\mathcal{K}$ as the set of users for which $\alpha_k>0$, and $\mathcal{P}$ must then be solved for users $k\in\mathcal{K}$. Since $\mathcal{P}$ is not a convex optimization problem due to the PU rate constraint, we first rewrite this constraint as follows. By regarding the aggregate interference as one interference signal with power $\sum_{k\in\mathcal{K}}a_kp_k$ and circularity coefficient $\sum_{k\in\mathcal{K}}a_kp_k\cc_k/\sum_{k\in\mathcal{K}}a_kp_k$, it can be noticed that the PU rate is only a function of the total interference power and the circularity coefficient of the aggregate interference. The rate constraint then establishes a relationship between both. Additionally, the PU rate $R$ decreases with the interference power and increases with its circularity coefficient. Thus, we have
\begin{equation}
	R\geq\bar{R} \; \Leftrightarrow \; \left\{\begin{matrix}
	\sum_{k\in\mathcal{K}}a_kp_k\leq t(\cc;\rhob) \\
	\sum_{k\in\mathcal{K}}a_kp_k\cc_k\geq t(\cc;\rhob)\cc
	\end{matrix}\right. \; ,
\end{equation}
where $t(\cc;\rhob)$ is the tolerable interference power when the circularity coefficient of the aggregate interference is $\cc$ (so that $R\geq\bar{R}$ is fulfilled), and $\rhob=[p_{\text{N}},\tilde{p}_{\text{N}}]$ is a vector of parameters, where $p_{\text{N}}$ and $\tilde{p}_{\text{N}}$ are the variance and complementary variance, respectively, of the noise at the primary receiver. Specifically, this function can be obtained in closed form by Corollary \ref{th:th11} as $t(\cc;\rhob)=a_{\text{S}}q(\cc)$ (i.e., from \eqref{eq:psqs} replacing $\cc_{\text{S}}$ with $\cc$ and $a_{\text{S}}p_{\text{S}}$ with $t(\cc;\rhob)$), taking $p_{\text{I}}=0$, and $\cc_{\text{I}}=0$. Thus, $\rhob=[p_{\text{N}}=1,\tilde{p}_{\text{N}}=0]$. Finally, the optimal value of $r$ in problem $\mathcal{P}$ can be found via bisection. That is, $r$ is fixed and $\mathcal{P}$ is thereby transformed into a feasibility problem. With this in mind, the $\ell$th step of the bisection method solves the feasibility problem 
\begin{align}
\mathcal{P}_\ell:\hspace{0.5cm} & \text{find}
& & \{p_k,\cc_k\}_{k\in\mathcal{K}},\cc \; ,\notag\\
& \text{subject to}
& & 0\leq p_k \leq P_k \; , \; k\in\mathcal{K} \; , \notag\\
& & & 0\leq\cc_k\leq1 \; , \; k\in\mathcal{K} \; ,\notag\\
& & & R_k(p_k,\cc_k)\geq\alpha_kr \; , \; k\in\mathcal{K} \; ,\notag \\
& & & \sum_{k\in\mathcal{K}}a_kp_k\leq t(\cc;\rhob) \; , \notag\\
& & & \sum_{k\in\mathcal{K}}a_kp_k\cc_k\geq t(\cc;\rhob)\cc \; . \notag
\end{align}

In the following, we will analyze this feasibility problem for a fixed value of $\cc$. This way, we will be able to obtain closed-form expressions for the feasible points in terms of $\cc$, which will provide insights into its optimal value. Let $\mathcal{K}_{\text{f}}$ denote the set of users for which $\log_2(1+P_k)=\alpha_kr$ (if any) and $\mathcal{K}_{\text{a}}=\mathcal{K}\backslash\mathcal{K}_{\text{f}}$. Notice that for users $k\in\mathcal{K}_{\text{f}}$ the rate $\alpha_kr$ is only achievable for $\cc_k=0$ and $p_k=P_k$. Therefore, the feasibility problem $\mathcal{P}_\ell$ is solved only for the users $k\in\mathcal{K}_{\text{a}}$, while treating the users in $\mathcal{K}_{\text{f}}$ as a fixed additional proper interference affecting the primary receiver, so that the equivalent noise level is $1+\sum_{k\in\mathcal{K}_{\text{f}}}a_kP_k$. Taking this into account, we may write the PU rate constraint as 
\begin{align}
	\sum_{k\in\mathcal{K}_{\text{a}}}a_kp_k&\leq t(\cc;\rhob) \; , \label{eq:teq}\\
    \sum_{k\in\mathcal{K}_{\text{a}}}a_kp_k\cc_k&\geq t(\cc;\rhob)\cc \; , \label{eq:keq}
\end{align}
where now $\rhob=[p_{\text{N}}=1+\sum_{k\in\mathcal{K}_{\text{f}}}a_kP_k,\tilde{p}_{\text{N}}=0]$. In order to obtain the transmission parameters that attain a given boundary point, we will first determine whether or not IGS is required to achieve that particular point. To this end, it is sufficient to analyze problem $\mathcal{P}_\ell$, at this stage, for $\cc\rightarrow0$. If $\mathcal{P}_\ell$ is feasible, there exists a feasible point such that \eqref{eq:teq} and \eqref{eq:keq} hold with equality while $p_k<P_k$ and $\cc_k<1$ for all $k\in\mathcal{K}_{\text{a}}$, and therefore these constraints can be removed from $\mathcal{P}_{\ell}$. This can be explained as follows. First, as $\cc$ approaches zero, every $\cc_k$ goes towards zero as well, and $\cc_k<1$ is then fulfilled for some feasible point (as the SU rates decrease with $\cc_k$). Second, as $\cc_k\rightarrow0$, the SU rates converge to $\log_2(1+p_k)$. Since the users in $\mathcal{K}_{\text{a}}$ fulfill $\log_2(1+P_k)>\alpha_kr$, a feasible point satisfying $p_k<P_k$ and \eqref{eq:teq} with equality can then be chosen. Obviously, this is the case as long as $t(\cc;\rhob)<\sum_{k\in\mathcal{K}_{\text{a}}}a_kP_k$ as $\cc\rightarrow0$. Otherwise, PGS is the optimal strategy for the whole rate region boundary as the PU rate constraint is inactive.  

For a given $\cc$, the feasibility problem $\mathcal{P}_\ell$ is a convex optimization problem (notice that $p_k\cc_k$ can be replaced by $\tilde{p}_k$ plus the constraint $0\leq\tilde{p}_k\leq p_k$). The Lagrangian reads
\begin{align}
		\mathcal{L}=&\sum_{k\in\mathcal{K}_{\text{a}}}\nu_k\left[R_k(p_k,c_k)-\alpha_kr\right]+\lambda\left(t(\cc;\rhob)-\sum_{k\in\mathcal{K}_{\text{a}}}a_kp_k\right)\notag\\&+\mu\left(\sum_{k\in\mathcal{K}_{\text{a}}}a_kp_k\cc_k-t(\cc;\rhob)\cc\right) \; ,\label{eq:L}
\end{align}
where we have dropped the constraints $p_k\leq P_k$ and $\cc_k\leq1$, as previously explained. In the above expression, $\nu_k\geq0$, $\lambda\geq0$ and $\mu\geq0$ are the Lagrange multipliers of the SU rate, interference power, and interference circularity coefficient constraints, respectively. Notice also that the constraints $p_k\geq0$ and $\cc_k\geq0$ have also been removed. As we will show later, these constraints will automatically be satisfied without considering them explicitly. Differentiating the Lagrangian with respect to $p_k$ and $\cc_k$ and equating them to zero yields
	\begin{align}
		\frac{a_k(p_k+1)}{\nu_k}&=\frac{1}{\lambda\left(1-\frac{\mu^2}{\lambda^2}\right)} \; , \label{eq:ap}\\
		\frac{a_kp_k\cc_k}{\nu_k}&=\frac{\mu}{\lambda^2\left(1-\frac{\mu^2}{\lambda^2}\right)} \; .\label{eq:apk}
	\end{align}
	Using these expressions we obtain
	\begin{align}
		\sum_{k\in\mathcal{K}_{\text{a}}}a_kp_k&=\frac{\sum_{k\in\mathcal{K}_{\text{a}}}\nu_k}{\lambda\left(1-\frac{\mu^2}{\lambda^2}\right)}-\sum_{k\in\mathcal{K}_{\text{a}}}a_k\leq t(\cc;\rhob) \; , \\
        \sum_{k\in\mathcal{K}_{\text{a}}}a_kp_k\cc_k&=\frac{\mu\sum_{k\in\mathcal{K}_{\text{a}}}\nu_k}{\lambda^2\left(1-\frac{\mu^2}{\lambda^2}\right)}\geq t(\cc;\rhob)\cc \; .
	\end{align}
	As mentioned before, if $\mathcal{P}_\ell$ is feasible, there is a feasible solution satisfying the above constraints with equality. We therefore take equality in the above constraints, and $\mathcal{P}_\ell$ will then be feasible if this solution also satisfies the rate constraints $R_k(p_k,c_k)\geq\alpha_kr$. By taking equality in the above expressions we obtain
\begin{align}
		\lambda&=\frac{\sum_{k\in\mathcal{K}_{\text{a}}}\nu_k}{\left(t(\cc;\rhob)+\sum_{k\in\mathcal{K}_{\text{a}}}a_k\right)\left[1-\frac{t(\cc;\rhob)^2\cc^2}{\left(t(\cc;\rhob)+\sum_{k\in\mathcal{K}_{\text{a}}}a_k\right)^2}\right]} \; ,\label{eq:l}\\
		\mu&=\frac{t(\cc;\rhob)\cc\sum_{k\in\mathcal{K}_{\text{a}}}\nu_k}{\left(t(\cc;\rhob)+\sum_{k\in\mathcal{K}_{\text{a}}}a_k\right)^2\left[1-\frac{t(\cc;\rhob)^2\cc^2}{\left(t(\cc;\rhob)+\sum_{k\in\mathcal{K}_{\text{a}}}a_k\right)^2}\right]} \; .\label{eq:m}
	\end{align}
From the derivatives of the Lagrangian and using the foregoing expressions we obtain
	\begin{align}
		&R_k(p_k,c_k)=\frac{1}{2}\log_2\left(\frac{p_k\cc_k\nu_k}{a_k\mu}\right)=\frac{1}{2}\log_2\left(\frac{a_kp_k\cc_k}{\nu_k}\frac{\nu_k^2}{a_k^2\mu}\right)\notag\\&=\frac{1}{2}\log_2\left[\frac{\nu_k^2}{a_k^2\left(\lambda^2-\mu^2\right)}\right]\notag\\&=\frac{1}{2}\log_2\left\{1+\frac{t(\cc;\rhob)}{\sum_{k\in\mathcal{K}_{\text{a}}}a_k}\left[\frac{t(\cc;\rhob)}{\sum_{k\in\mathcal{K}_{\text{a}}}a_k}\left(1-\cc^2\right)+2\right]\right\}\notag\\&+\log_2\frac{\nu_k\sum_{k\in\mathcal{K}_{\text{a}}}a_k}{a_k\sum_{k\in\mathcal{K}_{\text{a}}}\nu_k} \; .\label{eq:Rk}
	\end{align}
	That is, the point that satisfies the interference power and impropriety constraints with equality provides the SUs with the rate given by \eqref{eq:Rk}. As mentioned before, $\mathcal{P}_\ell$ is then feasible if there exist $\nu_k$, $k\in\mathcal{K}_{\text{a}}$, such that $R_k(p_k,c_k)\geq\alpha_kr$, which implies
	\begin{align}
		&\frac{\nu_k}{\sum_{k\in\mathcal{K}_{\text{a}}}\nu_k}\geq\notag\\&\resizebox{0.49\textwidth}{!}{$\frac{2^{\alpha_kr}a_k}{\left\{1+\frac{t(\cc;\rhob)}{\sum_{k\in\mathcal{K}_{\text{a}}}a_k}\left[\frac{t(\cc;\rhob)}{\sum_{k\in\mathcal{K}_{\text{a}}}a_k}\left(1-\cc^2\right)+2\right]\right\}^{1/2}\sum_{k\in\mathcal{K}_{\text{a}}}a_k} \; , \; k\in\mathcal{K}_{\text{a}} \; .$}\label{eq:nu}
	\end{align}
	The left-hand side of this expression is equal to or smaller than one, hence this condition is equivalent to
	\begin{align}
		\frac{1}{2}\log_2&\left\{1+\frac{t(\cc;\rhob)}{\sum_{k\in\mathcal{K}_{\text{a}}}a_k}\left[\frac{t(\cc;\rhob)}{\sum_{k\in\mathcal{K}_{\text{a}}}a_k}\left(1-\cc^2\right)+2\right]\right\}\notag\\&\geq\alpha_kr+\log_2\frac{a_k}{\sum_{k\in\mathcal{K}_{\text{a}}}a_k} \; , \; k\in\mathcal{K}_{\text{a}} \; .\label{eq:eq}
	\end{align}
  Notice that the above expression can be interpreted in a very insightful way. Its left-hand side is the rate of an equivalent user whose interference channel coefficient is the sum of the interference channels of all users, transmitting with a circularity coefficient $\cc$ and causing interference power $t(\cc;\rhob)$. This single-user representation of the SMAC is illustrated in Fig. \ref{fig:EqSUi}. Therefore, $r$ can be maximized by choosing the value of $\cc$ that maximizes the left-hand side of \eqref{eq:eq}. Before going any further, we first show that the transmit powers and circularity coefficients leading to \eqref{eq:eq} (for feasible $r$) are valid as they are non-negative. To this end, we combine \eqref{eq:ap}--\eqref{eq:m} and obtain
  \begin{align}
  	p_k&=\frac{\nu_k'\left(t(\cc;\rhob)+\sum_{k\in\mathcal{K}_{\text{a}}}a_k\right)}{a_k}-1 \; ,\label{eq:pk}\\
  	\cc_k&=\frac{\nu_k't(\cc;\rhob)\cc}{\nu_k'\left(t(\cc;\rhob)+\sum_{k\in\mathcal{K}_{\text{a}}}a_k\right)-a_k} \; ,\label{eq:kk}
  \end{align}
  where $\nu_k'=\frac{\nu_k}{\sum_{k\in\mathcal{K}_{\text{a}}}\nu_k}$. The transmit powers and circularity coefficients are non-negative if
  \begin{equation}
  	\nu_k'\geq\frac{a_k}{t(\cc;\rhob)+\sum_{k\in\mathcal{K}_{\text{a}}}a_k} \; .
  \end{equation}
  Combining the above expression with \eqref{eq:nu}, it can easily be seen that this is always the case.
\begin{figure}[t!]
\centering
\subfloat[]{\includegraphics[width=1\columnwidth]{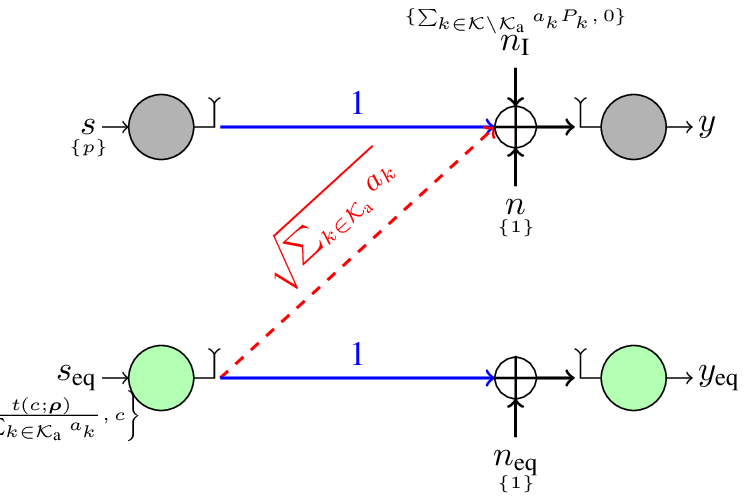}\label{fig:EqSUi}}\\%
\subfloat[]{\includegraphics[width=1\columnwidth]{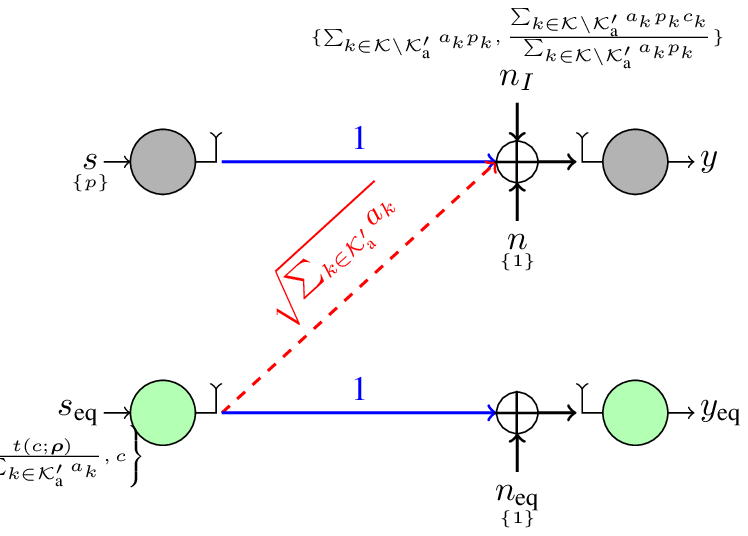}\label{fig:EqSUa}}%
\caption{Equivalent single-user canonical model of the SMAC. There are two possible situations depending on $\cc$: (a) every user in $\mathcal{K}_{\text{a}}$ fulfills $\cc_k<1$ and $p_k
<P_k$, or (b) one of these constraints is active for some users.}%
\label{fig:EqSU}%
\end{figure}  
  
  Since the foregoing analysis transforms the scenario into an equivalent single-user model, we can use our previous results for the single-SU case \cite{LameiroSantamariaSchreier:2015:Benefits-of-Improper-Signaling-for-Underlay} to obtain the following main result.
 \begin{theorem}\label{th:th0}
 	The boundary point characterized by $R_k=\alpha_kr^\star$, $k=1,\ldots,K$, with $r^\star$ being the optimal solution of $\mathcal{P}$, requires IGS if and only if
 	\begin{equation}
 		\sum_{k\in\mathcal{K}_{\text{\normalfont a}}}a_k\geq\sum_{k\in\mathcal{K}\backslash\mathcal{K}_{\text{\normalfont a}}}a_kP_k+\beta \; ,\label{eq:th0}
 	\end{equation}
 	where $\mathcal{K}$ and $\mathcal{K}_{\text{\normalfont a}}$ are the set of users for which $\alpha_k>0$ and $\log_2(1+P_k)>\alpha_kr^\star$, respectively, and $\beta$ is given by \eqref{eq:beta}. 
 \end{theorem}
 \begin{proof}
 	The result follows directly from \eqref{eq:eq} and Theorem 1 in \cite{LameiroSantamariaSchreier:2015:Benefits-of-Improper-Signaling-for-Underlay}. 
 \end{proof}
 \begin{corollary}
 	If the power budget is sufficiently high, the boundary point characterized by $R_k=\alpha_kr^\star$ requires IGS if and only if
 	\begin{equation}
 		\sum_{k\in\mathcal{K}}a_k\geq\beta \; .
 	\end{equation}
 \end{corollary}
 Theorem \ref{th:th0} provides a necessary and sufficient condition when IGS is required to achieve a given boundary point of the rate region. As explained earlier, this condition is also sufficient for the optimality of IGS (but not necessary) when the BS uses an MMSE-SIC receiver. In order to obtain the optimal transmission parameters that achieve that point, we have to analyze how the rate of the SUs changes as $\cc$, which is the circularity coefficient of the aggregate interference (caused by users in $\mathcal{K}_{\text{a}}$), increases. For the single-user case we state in \cite{LameiroSantamariaSchreier:2015:Benefits-of-Improper-Signaling-for-Underlay} that, if IGS is optimal, the rate increases monotonically with the circularity coefficient as long as the allowable transmit power is below the power budget. In our present scenario, there are multiple users each with their own power budget constraint and circularity coefficient. Therefore, as $\cc$ increases, we might reach a point where $p_k=P_k$ or $\cc_k=1$, for some $k\in\mathcal{K}_{\text{a}}$. In such a case, the behavior of the SU rates as $\cc$ increases beyond that point is no longer given by \eqref{eq:Rk}, since that expression is obtained assuming $p_k<P_k$ and $\cc_k<1$. Let $\cc=\cc'$ be the minimum value of $\cc$ such that there is one user $k'$ for which either $p_{k'}=P_{k'}$ or $\cc_{k'}=1$. Since the rate of this user must be at least $\alpha_{k'}r$, the other parameter is therefore given through $R_{k'}(P_{k'},\cc_{k'})=\alpha_{k'}r$ or $R_{k'}(p_{k'},1)=\alpha_{k'}r$, depending on the active constraint. In turn, the parameters of this user are known and the feasibility problem $\mathcal{P}_\ell$ must be solved for $k\in\mathcal{K}_{\text{a}}\backslash\{k'\}$. Therefore, we can regard the scenario in this case as one with users $k\in\mathcal{K}_{\text{a}}\backslash\{k'\}$ and a PU affected by an improper noise (due to user $k'$). This equivalent single-user representation is depicted in Fig. \ref{fig:EqSUa}, where $\mathcal{K}'_{\text{a}}=\mathcal{K}_{\text{a}}\backslash\{k'\}$ and the parameter vector is now $\rhob=[p_{\text{N}}=1+\sum_{k\in\mathcal{K}\backslash\mathcal{K}'_{\text{a}}}a_kp_k,\tilde{p}_{\text{N}}=\sum_{k\in\mathcal{K}\backslash\mathcal{K}'_{\text{a}}}a_kp_k\cc_k]$.  
 
Notice that this only implies a small change in our analysis of $\mathcal{P}_\ell$ in \eqref{eq:teq}--\eqref{eq:eq}. First, the set of users $\mathcal{K}_{\text{a}}$ is replaced by $\mathcal{K}_{\text{a}}\backslash\{k'\}$, Second, the parameter vector $\rhob$ of the tolerated interference power $t(\cc;\rhob)$ is now different, as given above. Nevertheless, this equivalent scenario follows the model considered in Section \ref{sec:ii}, thus \added{the behavior of the system is now governed by the results obtained in Section \ref{sec:ii}, so that the analysis can be further carried out by applying Theorem 1}. Additionally, $t(\cc;\rhob)$ can be obtained in closed form by Corollary \ref{th:th11} taking $p_{\text{I}}=p_{\text{N}}-1$ and $\cc_{\text{I}}=\frac{\tilde{p}_{\text{N}}}{p_{\text{N}}-1}$.

In order to identify the value of $\cc$, $c'$, such that one user reaches its power budget or its circularity coefficient goes to 1, we take equality in \eqref{eq:nu} and plug it in \eqref{eq:pk} and \eqref{eq:kk} yielding
\begin{align}
	p_k&=2^{\alpha_kr}\left(\frac{t(\cc;\rhob)}{\sum_{k\in\mathcal{K}_{\text{a}}}a_k}+1\right)\times\notag\\&\left\{1+\frac{t(\cc;\rhob)}{\sum_{k\in\mathcal{K}_{\text{a}}}a_k}\left[\frac{t(\cc;\rhob)}{\sum_{k\in\mathcal{K}_{\text{a}}}a_k}\left(1-\cc^2\right)+2\right]\right\}^{-1/2}-1 \; ,\label{eq:pk2}\\
	\cc_k&=\frac{2^{\alpha_kr}t(\cc;\rhob)\cc}{2^{\alpha_kr}\left(t(\cc;\rhob)+\sum_{k\in\mathcal{K}_{\text{a}}}a_k\right)}\times\notag\\&\frac{\left\{1+\frac{t(\cc;\rhob)}{\sum_{k\in\mathcal{K}_{\text{a}}}a_k}\left[\frac{t(\cc;\rhob)}{\sum_{k\in\mathcal{K}_{\text{a}}}a_k}\left(1-\cc^2\right)+2\right]\right\}^{-1/2}}{\left\{1+\frac{t(\cc;\rhob)}{\sum_{k\in\mathcal{K}_{\text{a}}}a_k}\left[\frac{t(\cc;\rhob)}{\sum_{k\in\mathcal{K}_{\text{a}}}a_k}\left(1-\cc^2\right)+2\right]\right\}^{-1/2}-1} \; .\label{eq:kk2}
\end{align}
By looking into these expressions we notice the following. Let 
\begin{align}
	k_{\text{p}}&=\underset{k}{\arg\min}\,\frac{P_k-1}{2^{\alpha_kr}} , & k_{\text{c}}&=\underset{k}{\arg\min}\,\alpha_k \; ,\label{eq:kp}\\
	\cc_{\text{p}}&=\underset{\cc}{\arg\min}\,\{p_{k_{\text{p}}}=P_{k_{\text{p}}}\}  , & \cc_{\text{c}}&=\underset{\cc}{\arg\min}\,\{\cc_{k_{\text{c}}}=1\} \; ,\label{eq:kappakappa}
\end{align}
with $p_k$ and $\cc_k$ respectively given by \eqref{eq:pk2} and \eqref{eq:kk2}. Then the SU rates are given by \eqref{eq:Rk} for $0\leq\cc\leq\cc'=\min(\cc_{\text{p}},\cc_{\text{c}})$. As mentioned before, when $\cc$ increases beyond that point, user $k_{\text{p}}$ or $k_{\text{c}}$ has known parameters. Denoting this user as $k'$, the optimization is carried out for users $k\in\mathcal{K}_{\text{a}}\backslash\{k'\}$, and expression \eqref{eq:Rk} is valid by replacing $\mathcal{K}_{\text{a}}$ with $\mathcal{K}_{\text{a}}'=\mathcal{K}_{\text{a}}\backslash\{k'\}$ and updating $\rhob$ accordingly, since now the PU is affected by improper noise in the equivalent single-user representation (due to user $k'$ having fixed and known parameters, as depicted in Fig. \ref{fig:EqSUa}). However, the updated version of \eqref{eq:Rk} will only be valid as long as the remaining users in set $\mathcal{K}_{\text{a}}'=\mathcal{K}_{\text{a}}\backslash\{k'\}$ have a transmit power smaller than their power budget and circularity coefficients smaller than one. Hence, if the SU rates still increase beyond that point, we go back again to identifying the value of $\cc$ at which the next user reaches one of these bounds (if there is such a user) and the aforementioned procedure is repeated again until the optimal value of $\cc$ is identified. We describe this procedure in Algorithm \ref{alg:pl}, where we assume, for the sake of illustration, that $\alpha_k>0$ and $\log_2(1+P_k)>\alpha_kr$ $\forall k$. Notice that the loop in the algorithm successively identifies the users for which either $p_k=P_k$ or $\cc_k=1$ is optimal. Therefore, the maximum number of iterations is equal to the number of SUs, $K$. Furthermore, the computational complexity is small as every step involves only closed-form expressions.

\begin{algorithm}[t!]
\begin{algorithmic}
\STATE{Set $\mathcal{K}_{\text{a}}=\{1,\ldots,K\}$ and $\rhob=[p_{\text{N}}=1,\tilde{p}_{\text{N}}=0]$.}
\LOOP
    \STATE{Use Theorem \ref{th:th1} to obtain $\cc_{\text{S}}^\star$, by replacing $a_{\text{S}}$ with $\sum_{k\in\mathcal{K}_{\text{a}}}a_k$ and taking $p_{\text{I}}=p_{\text{N}}-1$, $\cc_{\text{I}}=\frac{\tilde{p}_{\text{N}}}{p_{\text{N}}-1}$, and $P_{\text{S}}=\sum_{k\in\mathcal{K}_{\text{a}}}P_k$.}
    \STATE{Use Corollary \ref{th:th11} to obtain $t(\cc;\rhob)=a_{\text{S}}q(\cc)$.}
    \STATE{Obtain $k_{\text{p}}$, $k_c$, $\cc_{\text{p}}$, and $\cc_{\text{c}}$ by means of \eqref{eq:kp}--\eqref{eq:kappakappa}, and set $\cc'=\min(\cc_{\text{p}},\cc_{\text{c}})$.}
    \IF{$\cc_{\text{S}}^\star\leq\cc'$}
    \STATE{Set $\cc=\cc_{\text{S}}^\star$ and determine feasibility by checking \eqref{eq:eq} $\forall k$.}
    \STATE{If the problem is feasible obtain $p_k$ and $\cc_k$ $\forall k\in\mathcal{K}_a$ by \eqref{eq:pk2} and \eqref{eq:kk2} and stop.}
    \ELSE
    \STATE{Set $k'=k_{\text{p}}$ if $\cc_{\text{p}}<\cc_{\text{c}}$ and $k'=k_{\text{c}}$ otherwise, and $\cc=\cc'$.}
    \STATE{Obtain $p_{k'}$ and $\cc_{k'}$ by \eqref{eq:pk2} and \eqref{eq:kk2}, and update $\mathcal{K}_{\text{a}}=\mathcal{K}_{\text{a}}\backslash\{k'\}$, $p_{\text{N}}=p_{\text{N}}+a_{k'}p_{k'}$, and $\tilde{p}_{\text{N}}=\tilde{p}_{\text{N}}+a_{k'}p_{k'}\cc_{k'}$.}
    \ENDIF
\ENDLOOP
\end{algorithmic}
\caption{Algorithm to solve $\mathcal{P}_\ell$.} \label{alg:pl}
\end{algorithm}

The whole procedure to solve the initial problem $\mathcal{P}$ is the following: The maximum value of $r$ is found via bisection method. For each step of the bisection method, $\mathcal{P}_\ell$ is solved by Algorithm \ref{alg:pl}. Upon convergence, the solution of $\mathcal{P}_\ell$ is unique and provides the optimal transmission parameters that attain the point of the rate region boundary.

\begin{figure}[t!]
\centering
\includegraphics[width=1\columnwidth]{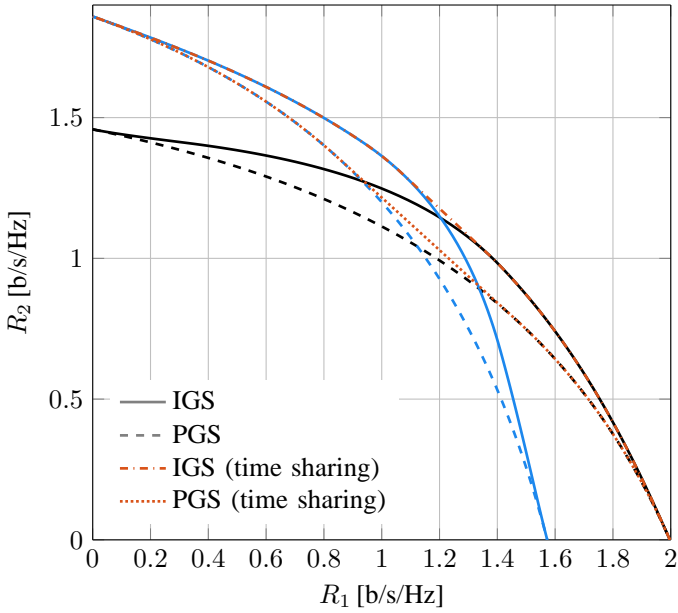}
\caption{Rate region boundaries. We consider $a_1<\beta$, $a_2<\beta$ and $a_1+a_2>\beta$ for both user orderings. We also depict the boundaries when time sharing is allowed.}
\label{fig:BoundaryEx1}
\end{figure}
\begin{figure}[t!]
\centering
\includegraphics[width=1\columnwidth]{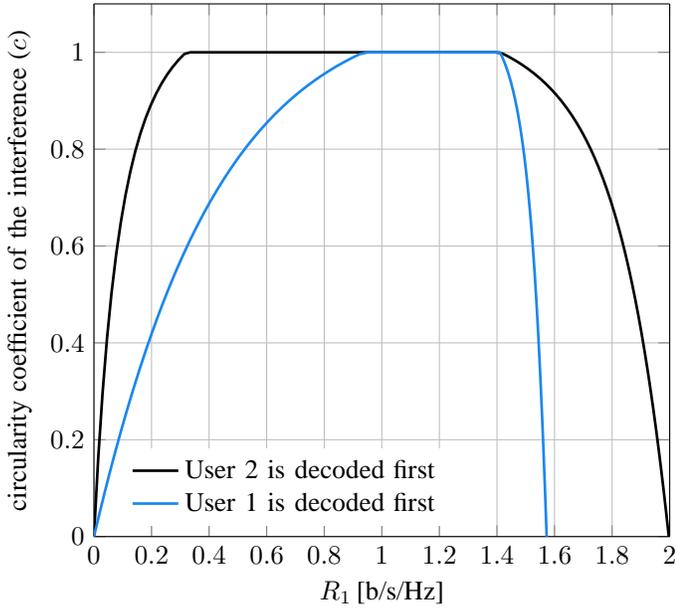}
\caption{Optimal interference circularity coefficient for the rate region boundaries in Fig. \ref{fig:BoundaryEx1}.}
\label{fig:KappaEx1}
\end{figure}

\section{Numerical examples}\label{sec:sim}
\begin{figure}[t!]
\centering
\includegraphics[width=1\columnwidth]{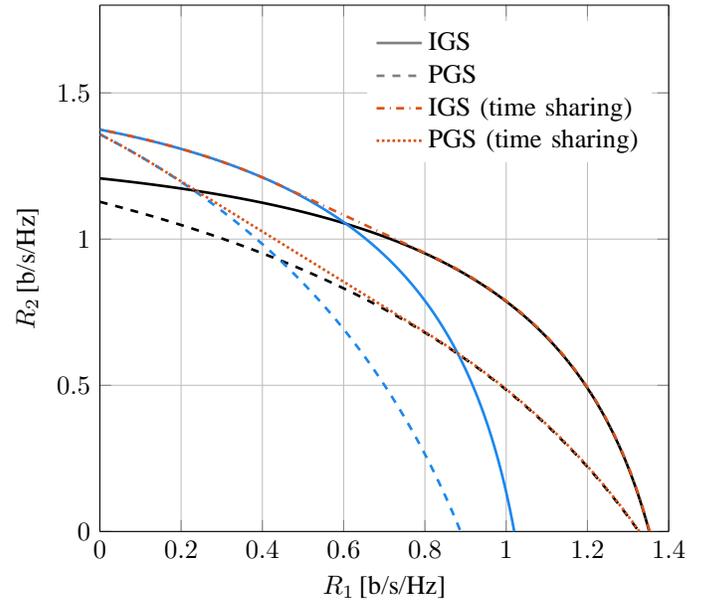}
\caption{Rate region boundaries. We consider $a_1>\beta$ and $a_2>\beta$ for both user orderings. We also depict the boundaries when time sharing is allowed.}
\label{fig:BoundaryEx2}
\end{figure}
In this section we present simulation examples to illustrate the performance improvements achieved by IGS over its proper counterpart.  We start by considering $K=2$ SUs, so that the achievable rate region for the secondary network can easily be depicted in a figure. We set the power budgets in the original model to $P'_1=P'_2=p'=100$, $\sigma^2=1$, and $\bar{R}$ as 80\% of the PU capacity. First, we consider a scenario with channels
\begin{align}
	&\begin{bmatrix}
	\h_1 & \h_2
	\end{bmatrix}=\begin{bmatrix}
	   2.6366 - j0.3382 & -2.8824 - j0.1728 \\
  -1.4428 + j1.0861 & -1.7887 + j2.0730
	\end{bmatrix} \; , \\ &\begin{bmatrix} h \\ g_1 \\ g_2\end{bmatrix}=\begin{bmatrix}-0.8815 + j0.4721 \\ 0.0533 + j0.2217 \\ 0.2221 + j0.1991\end{bmatrix} \; , \; \g=\begin{bmatrix}
		0.0533 + j0.2217 \\
   0.2221 + j0.1991
	\end{bmatrix} \; . 
\end{align}
\begin{figure}[t!]
\centering
\includegraphics[width=1\columnwidth]{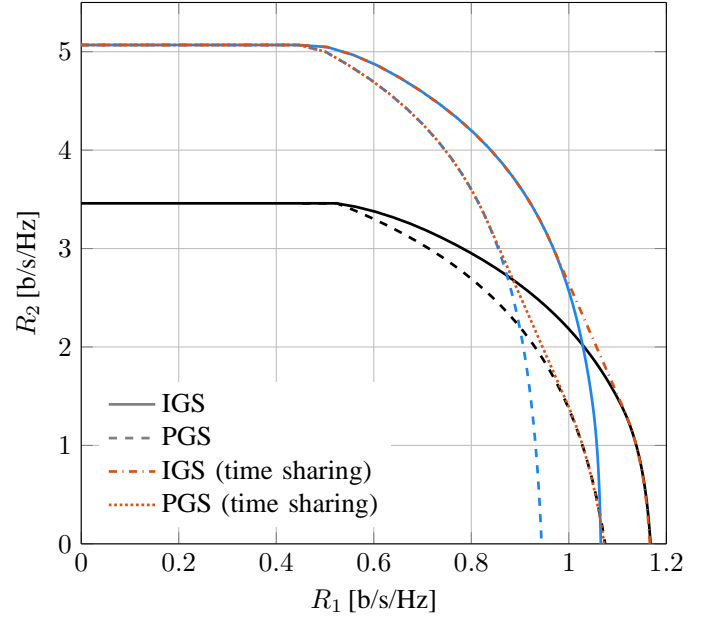}
\caption{Rate region boundaries. We consider $a_1>\beta$ and $a_2\ll\beta$ for both user orderings. We also depict the boundaries when time sharing is allowed.}
\label{fig:BoundaryEx3}
\end{figure}
These channel coefficients yield $\bar{R}=5.33$ b/s/Hz and $\beta=0.94$. The interference channel coefficients of the canonical model associated to this scenario are $a_1=0.52$ and $a_2=0.89$ when user 2 is decoded first; and $a_1=0.788$ and $a_2=0.592$ otherwise. In both cases we have $a_1+a_2>\beta$. According to Theorem \ref{th:th0}, IGS is optimal for $0<\alpha_1<1$ for both user orderings, whereas PGS is optimal for the extreme points $\alpha_1=0$ and $\alpha_1=1$ (as both $a_1$ and $a_2$ are smaller than $\beta$). The boundary of the rate regions for the optimal strategy that permits IGS (obtained by Algorithm \ref{alg:pl} and bisection method for $r$) is depicted in Fig. \ref{fig:BoundaryEx1}, along with that achieved by PGS. We also depict the rate region boundaries achieved by time sharing, which results in the convex hull of the rate regions corresponding to each user ordering.\footnote{We note that time sharing yields the convex hull when the power constraint is satisfied in each operation point. The rate region may be enlarged by constraining the average transmit power over the different operation points instead \cite{Hellings17}. Such an analysis is, however, not straight forward and falls outside the scope of this paper.} As expected, the boundaries achieved by PGS and IGS overlap at the extreme points. However, the rate region achieved by IGS is significantly larger than that obtained by PGS for both user orderings, which is especially noticeable for intermediate values of $R_1$. As can also be observed, time sharing keeps the gap between both signaling schemes almost unchanged. Fig. \ref{fig:KappaEx1} shows the optimal circularity coefficient of the aggregate interference, $\cc$, that achieves the optimal rate region boundary. We observe that these intermediate $R_1$ values, for which the IGS rate region is significantly larger than the proper one, correspond to maximally IGS, i.e., $\cc_1=1$ and $\cc_2=1$. 

For the second scenario we consider the channel coefficients
\begin{align}
	&\begin{bmatrix}
	\h_1 & \h_2
	\end{bmatrix}=\begin{bmatrix}
	      0.1599 - j0.9812 &  1.0563 + j0.8070 \\
  -0.5172 + j0.4742  & 1.1759 + j0.9756
	\end{bmatrix} \; , \\ &\begin{bmatrix} h \\ g_1 \\ g_2\end{bmatrix}=\begin{bmatrix}0.9445 + j0.3284 \\ 0.2908 + j0.1358 \\
   0.3279 + j0.1532\end{bmatrix} \; , \; \g=\begin{bmatrix}
		  -0.3209 - j0.0052 \\
  -0.1427 - j0.3326
	\end{bmatrix} \; , 
\end{align}
which also yield $\bar{R}=5.33$ b/s/Hz and $\beta=0.94$. The interference channel coefficients in the canonical model are now $a_1=1.03$ and $a_2=1.31$ when the second user is decoded first; and $a_1=1.829$ and $a_2=0.995$ otherwise. As $a_1$ and $a_2$ are both greater than $\beta$, IGS is optimal on the whole rate region boundary. This is observed in Fig. \ref{fig:BoundaryEx2}. Again, we notice that the rate improvements are more prominent for intermediate rates than at the extreme points. As an example, let us consider the region obtained when the second user is decoded first (depicted in black). While the second user can achieve a rate around 7\% higher compared to PGS for $R_1=0$, the rate improvement is approximately 40\% for $R_1=0.8$ b/s/Hz. In this example, all boundary points are achieved by maximally IGS, i.e., $\cc_1=\cc_2=1$. 

As a final example for the 2-user case, we consider the channels
\begin{align}
	&\begin{bmatrix}
	\h_1 & \h_2
	\end{bmatrix}=\begin{bmatrix}
	      2.1257 - j3.0397 & -0.4956 + j0.9835 \\
   0.5401 - j0.9356  & 2.1329 - j0.6720
	\end{bmatrix} \; , \\ &\begin{bmatrix} h \\ g_1 \\ g_2\end{bmatrix}=\begin{bmatrix}0.8292 + j0.5589 \\ -0.0869 + j0.3653 \\
   0.0301 + j0.0900\end{bmatrix} \; , \; \g=\begin{bmatrix}
		  1.1379 + j0.7371 \\
   0.2219 - j0.2120
	\end{bmatrix} \; . 
\end{align}
For these coefficients we have again $\bar{R}=5.33$ b/s/Hz and $\beta=0.94$. In the canonical model, the interference channel coefficients are $a_1=1.41$ and $a_2=0.09$ when the second user is decoded first; and $a_1=1.684$ and $a_2=0.028$ otherwise. The achievable rate region boundaries are depicted in Fig. \ref{fig:BoundaryEx3}. Since $a_2$ is very small for both user orderings, the second user can use PGS with maximum power in a large portion of the boundary, from $R_1=0$ until approximately $R_1=0.5$ b/s/Hz. In that case, the optimality condition for IGS is not met, and the optimal strategy of the first user is also PGS. Nevertheless, the rate of user 2 eventually falls below its maximum value, and IGS again starts to provide significant gains in terms of achievable rate. Indeed, even though $a_2\ll\beta$, the rate of user 2 can be substantially enlarged by IGS. For example, when user 2 is decoded first and for $R_1=1$ b/s/Hz, user 2 can achieve a rate approximately 55\% higher compared to PGS. 
\begin{figure}[t!]
\centering
\includegraphics[width=1\columnwidth]{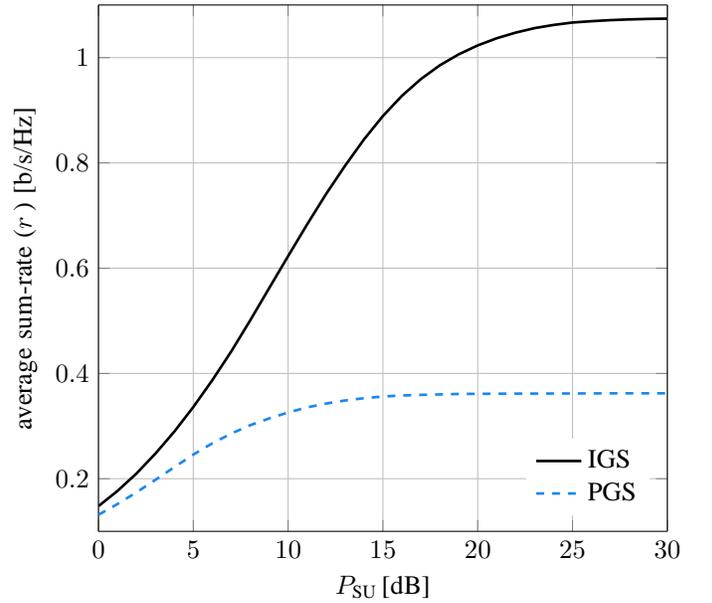}
\caption{Average sum-rate at an arbitrary point of the rate region boundary as a function of the SU power budget.}
\label{fig:rEx4}
\end{figure}

In the next scenario we consider $K=4$ and $N=4$. We consider an arbitrary point of the rate region boundary specified by $\boldsymbol{\alpha}=[0.27,0.13,0.09,0.51]$, and depict in Fig. \ref{fig:rEx4} the average sum-rate $r$ obtained after 1000 Monte Carlo simulations as a function of $P_1=P_2=\cdots=P_K=P_{\text{SU}}$. We consider Rayleigh fading, thus each channel coefficient has been independently drawn from a proper complex Gaussian distribution with zero mean and unit variance. For each channel realization, the PU rate constraint has been set as 60\% of its capacity in the absence of interference. We observe in Fig. \ref{fig:rEx4} that the improvement between IGS and PGS in terms of sum-rate grows with the power budget. This behavior is intuitively clear: as IGS permits increasing the transmit power, the increase will be larger the higher the available transmit power is. However, the performance is eventually limited by interference, thus the rate saturates as seen in Fig. \ref{fig:rEx4}. At the saturation point, we observe that IGS achieves almost three times as much sum-rate as PGS.
\begin{figure}[t!]
\centering
\includegraphics[width=1\columnwidth]{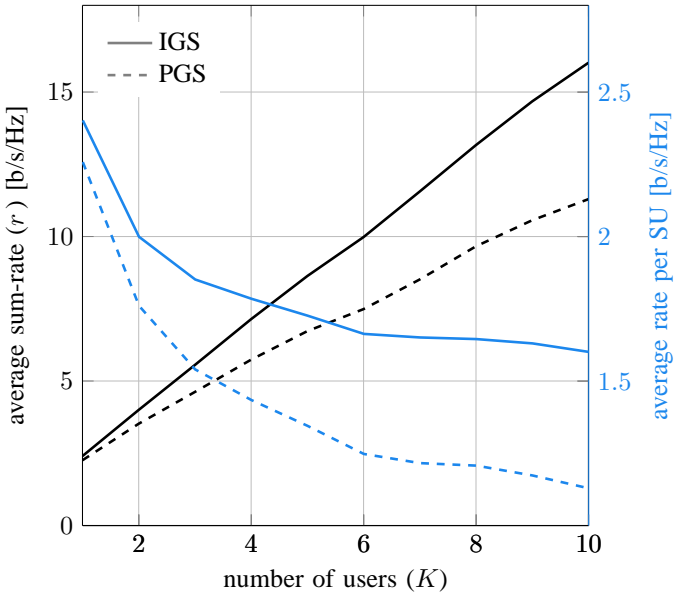}
\caption{Average sum-rate at the fairness point of the rate region boundary as a function of the number of users, with $N=K$.}
\label{fig:rEx6}
\end{figure}

As a final example, we evaluate the average sum-rate as the number of users $K$ grows, with $N=K$. Each channel coefficient follows a proper complex Gaussian distribution with zero mean and unit variance, except for the PU-SU channel $\g$, which is set to zero. We consider the point on the rate region boundary characterized by $\alpha_k=\frac{1}{K}$, $\forall k$. This is the fairness point as each SU attains the same rate. The results are depicted in Fig. \ref{fig:rEx6}, where both the average sum-rate and average rate per user are depicted. As observed, the gap in the sum-rate between PGS and IGS grows with the number of users. This is not only due to the fact that the single-user gap is scaled by a factor of $K$, but also because the rate per SU presents an increasing gap between IGS and PGS, as it is also shown in Fig. \ref{fig:rEx6}. This result indicates that IGS becomes more profitable as the number of users increases, showing its suitability for multiuser interference networks. 

\section{Conclusions}\label{sec:con}
In this paper we have analyzed IGS for an underlay MAC with ZF decoding, which shares the spectrum with a PU. Firstly, we have derived a necessary and sufficient condition for IGS to outperform PGS. Secondly, we have presented an efficient algorithm to compute every point of the boundary of the rate region. Several numerical examples have shown that IGS can achieve significantly higher rates than PGS, especially as the number of users increases. Our results provide performance limits of an underlay MAC with ZF decoding and can be used as design guidelines for arbitrary decoding schemes, when only partial CSI is available, and also outside the cognitive radio context if the point-to-point transmitter is restricted to PGS due to lack of CSI or because it is a legacy user.

\section*{Acknowledgments}
The work of C. Lameiro and P. J. Schreier was supported by the German Research Foundation (DFG) under grants SCHR 1384/6-1 and LA 4107/1-1. The work of I. Santamar{\'i}a was supported by the Ministerio de Econom{\'i}a y Competitividad (MINECO) and AEI/FEDER funds of the UE, Spain, under projects RACHEL (TEC2013-47141-C4-3-R) and CARMEN (TEC2016-75067-C4-4-R).

\appendices
	
\section{Proof of Lemma \ref{th:l1}}\label{app:l1}
We start presenting the following auxiliary lemma, that will be used in the proof.
\begin{lemma}\label{th:l0}
	$R_{\text{\normalfont S}}(\cc_{\text{\normalfont S}})$ is non-decreasing at $\cc_{\text{\normalfont S}}$ if and only if
	\begin{equation}\label{eq:drsuf}
		a_{\text{\normalfont S}}q(\cc_{\text{\normalfont S}})\cc_{\text{\normalfont S}}\left(1-\frac{p_{\text{\normalfont I}}+\beta}{a_{\text{\normalfont S}}}\right)+p_{\text{\normalfont I}}\cc_{\text{\normalfont I}}\left[1+q(\cc_{\text{\normalfont S}})\right]\geq 0 \; .
	\end{equation}
\end{lemma}
\begin{proof}
	Let us rewrite \eqref{eq:Rk0} as
	\begin{equation}
		2^{2R_{\text{S}}(\cc_{\text{S}})}-1=p_{\text{S}}(\cc_{\text{S}})^2(1-\cc_{\text{S}}^2)+2p_{\text{S}}(\cc_{\text{S}}) \; .
	\end{equation}
	Taking $p_{\text{S}}(\cc_{\text{S}})=q(\cc_{\text{S}})$, \eqref{eq:psqs} holds with equality. Therefore, we can replace $p_{\text{S}}(\cc_{\text{S}})^2(1-\cc_{\text{S}}^2)$ in the foregoing equation by  
	\begin{align}
	&p_{\text{S}}(\cc_{\text{S}})^2(1-\cc_{\text{S}}^2)=\frac{1}{a_{\text{S}}^2}\{(1-\beta)[p+2(1+a_{\text{S}}q(\cc_{\text{S}})+p_{\text{I}})]-1\notag\\&-2(a_{\text{S}}q(\cc_{\text{S}})+p_{\text{I}})-p_{\text{I}}^2(1-\cc_{\text{I}}^2)-2a_{\text{S}}q(\cc_{\text{S}})p_{\text{I}}(1-\cc_{\text{S}}\cc_{\text{I}})\} \; ,
	\end{align}
	which yields
	\begin{align}
		&2^{2R_{\text{S}}(\cc_{\text{S}})}-1=\frac{1}{a_{\text{S}}^2}\left\{(1-\beta)[p+2(1+a_{\text{S}}q(\cc_{\text{S}})+p_{\text{I}})]-1\right.\notag\\&-2(a_{\text{S}}q(\cc_{\text{S}})+p_{\text{I}})-p_{\text{I}}^2(1-\cc_{\text{I}}^2)\left.-2a_{\text{S}}q(\cc_{\text{S}})p_{\text{I}}(1-\cc_{\text{S}}\cc_{\text{I}})\right\}\notag\\&+2q(\cc_{\text{S}}) \; .
\end{align}
	To evaluate whether or not $R_{\text{S}}(\cc_{\text{S}})$ is increasing, we can alternatively consider $2^{2R_{\text{S}}(\cc_{\text{S}})}-1$. Taking this into account we obtain
	\begin{align}\label{eq:drsu}
	\frac{\partial R_{\text{S}}(\cc_{\text{S}})}{\partial \cc_{\text{S}}}\geq 0 \; \Leftrightarrow \; &\frac{\partial q(\cc_{\text{S}})}{\partial \cc_{\text{S}}}\left[1-\frac{p_{\text{I}}(1-\cc_{\text{S}}\cc_{\text{I}})+\beta}{a_{\text{S}}}\right]\notag\\&+\frac{q(\cc_{\text{S}})p_{\text{I}}\cc_{\text{I}}}{a_{\text{S}}}\geq 0 \; .
	\end{align}
	To obtain the derivative of $q(\cc_{\text{S}})$, we make use of \eqref{eq:psqs}. Replacing $p_{\text{S}}(\cc_{\text{S}})$ with $q(\cc_{\text{S}})$, \eqref{eq:psqs} holds with equality. We can then take the derivative of both sides of the equality with respect to $\cc_{\text{S}}$, which yields, after some manipulations, 
	\begin{equation}\label{eq:dq}
		\frac{\partial q(\cc_{\text{S}})}{\partial \cc_{\text{S}}}=\frac{q(\cc_{\text{S}})[a_{\text{S}}q(\cc_{\text{S}})\cc_{\text{S}}+p_{\text{I}}\cc_{\text{I}}]}{a_{\text{S}}q(\cc_{\text{S}})(1-\cc_{\text{S}}^2)+p_{\text{I}}(1-\cc_{\text{S}}\cc_{\text{I}})+\beta} \; .
	\end{equation}
	Combining \eqref{eq:dq} and \eqref{eq:drsu}, we obtain, after some manipulations,
	\begin{align}
		\frac{\partial R_{\text{S}}(\cc_{\text{S}})}{\partial \cc_{\text{S}}}\geq 0 \; \Leftrightarrow \; &a_{\text{S}}q(\cc_{\text{S}})\cc_{\text{S}}\left(1-\frac{p_{\text{I}}+\beta}{a}\right)\notag\\&+p_{\text{I}}\cc_{\text{I}}\left[1+q(\cc_{\text{S}})\right]\geq 0 \; .
	\end{align}
\end{proof}
To prove Lemma \ref{th:l1} we have to show that \eqref{eq:drsuf} is only satisfied for $0\leq\cc_{\text{S}}\leq\cc_{\text{R}}$ and some $\cc_{\text{R}}$. First, by Lemma \ref{th:l0}, the derivative is positive at $\cc_{\text{S}}=0$ if
	\begin{equation}
		\left.\frac{\partial R_{\text{S}}(\cc_{\text{S}})}{\partial \cc_{\text{S}}}\right|_{\cc_{\text{S}}=0}> 0 \; \Leftrightarrow \; p_{\text{I}}\cc_{\text{I}}\left[1+q(0)\right]> 0 \; ,
	\end{equation}
	which always holds since $q(0)\geq 0$, $p_{\text{I}}>0$ and $\cc_{\text{I}}>0$. To see what happens as $\cc_{\text{S}}$ increases, we may consider two cases. When $1-\frac{p_{\text{S}}+\beta}{a_{\text{S}}}\geq 0$, \eqref{eq:drsuf} is satisfied for $0\leq\cc_{\text{S}}\leq1$, since all the involved quantities are non-negative. Therefore, in this case, $R_{\text{S}}(\cc_{\text{S}})$ is increasing for all values of $\cc_{\text{S}}$. Now let $1-\frac{p_{\text{I}}+\beta}{a_{\text{S}}}<0$. Let us also rewrite \eqref{eq:drsuf} as
	\begin{align}\label{eq:drsuf2}
		\frac{\partial R_{\text{S}}(\cc_{\text{S}})}{\partial \cc_{\text{S}}}\geq 0 \; \Leftrightarrow \; &a_{\text{S}}q(\cc_{\text{S}})\left[\cc_{\text{S}}\left(1-\frac{p_{\text{I}}+\beta}{a_{\text{S}}}\right)+\frac{p_{\text{I}}\cc_{\text{I}}}{a_{\text{S}}}\right]\notag\\&+p_{\text{I}}\cc_{\text{I}}\geq 0 \; .
	\end{align}
	As $1-\frac{p_{\text{S}}+\beta}{a_{\text{S}}}<0$, the term that multiplies $q(\cc_{\text{S}})$ in the above expression decreases with $\cc_{\text{S}}$. It may then happen that this term becomes negative, in which case $a_{\text{S}}q(\cc_{\text{S}})\left[\cc_{\text{S}}\left(1-\frac{p_{\text{I}}+\beta}{a_{\text{S}}}\right)+\frac{p_{\text{I}}\cc_{\text{I}}}{a_{\text{S}}}\right]$ becomes a negative and decreasing function (because $q(\cc_{\text{S}})$ is increasing in $\cc_{\text{S}}$ as observed in \eqref{eq:dq}). Therefore, there may exist $\cc_{\text{R}}$ such that \eqref{eq:drsuf2} holds with equality for $\cc_{\text{S}}=\cc_{\text{R}}$, in which case \eqref{eq:drsuf2} will not hold for $\cc_{\text{S}}>\cc_{\text{R}}$. This concludes the proof.	
	
\section{Proof of Theorem \ref{th:th1}}\label{app:th1}
Let us start with the first circularity coefficient, $\cc_{\text{B}}$, associated with the power budget constraint. If $\cc_0$ is such that $q(\cc_0)=P_{\text{S}}$, then the transmit power of the SU is $p_{\text{S}}(\cc_{\text{S}})=P_{\text{S}}$ for $\cc_{\text{S}}\geq\cc_0$. Therefore, its achievable rate is decreasing in that interval as can be seen from \eqref{eq:Rk0}, and hence the optimal circularity coefficient $\cc_{\text{S}}^\star$ cannot be greater than $\cc_0$. Since it may happen that such a $\cc_0$ does not exist, i.e., $q(\cc_{\text{S}})<P_{\text{S}}$ or $q(\cc_{\text{S}})>P_{\text{S}}$ can be fulfilled for $0\leq\cc_{\text{S}}\leq1$, we obtain $\cc_{\text{S}}^\star\leq\cc_{\text{B}}$. Let us now consider the second circularity coefficient, $\cc_{\text{R}}$, which is related to the PU rate constraint. Since $\cc_{\text{B}}$ considers the power budget constraint, we drop this constraint at this point to obtain $\cc_{\text{R}}$. By Lemma \ref{th:l1} we know that, when $p(\cc_{\text{S}})=q(\cc_{\text{S}})$, $R_{\text{S}}(\cc_{\text{S}})$ achieves its maximum at the maximum value of $\cc_{\text{S}}$ for which $\frac{\partial R_{\text{S}}(\cc_{\text{S}})}{\partial \cc_{\text{S}}}\geq0$. Thus, $\cc_{\text{S}}=1$ is optimal if the foregoing derivative is non-negative at $\cc_{\text{S}}=1$. If this condition does not hold, there is, according to Lemma \ref{th:l1}, $\cc_{\text{R}}<1$ such that $\left.\frac{\partial R_{\text{S}}(\cc_{\text{S}})}{\partial \cc_{\text{S}}}\right|_{\cc_{\text{S}}=\cc_{\text{R}}}=0$, corresponding to the maximum of $R_{\text{S}}(\cc_{\text{S}})$. This yields \eqref{eq:kR}. Since $R_{\text{S}}(\cc_{\text{S}})$ is increasing in the interval $0\leq\cc_{\text{S}}\leq\cc_{\text{R}}$, we finally obtain \eqref{eq:kStar}.

\section{Proof of Corollary \ref{th:th11}}\label{app:th11}
The closed-from expression for $q(\cc_{\text{S}})$ is obtained by taking $p_{\text{S}}=q(\cc_{\text{S}})$ in \eqref{eq:psqs} with equality, which yields a second-order equation for $q(\cc_{\text{S}})$ and hence a closed-from expression.

To obtain a closed-form expression for $\cc_{\text{B}}$ we proceed as follows. First, since $q(\cc_{\text{S}})$ is an increasing function, by \eqref{eq:kP} we have $\cc_{\text{B}}=0$ for $P_{\text{S}}\leq q(0)$. Otherwise, assuming that $q(\cc_{\text{B}})=P_{\text{S}}$, we take  $p_{\text{S}}=P_{\text{S}}$ in \eqref{eq:psqs} with equality, which yields a second-order equation, thereby obtaining $\cc_{\text{B}}$ in closed-form. If the resulting solution is invalid, this means that $\cc_{\text{B}}=1$ and $q(\cc_{\text{B}})<P_{\text{S}}$.

To obtain the closed-form expression of $\cc_{\text{R}}$ we first notice that $\cc_{\text{R}}=1$ if the condition in \eqref{eq:kR} is fulfilled for $\cc_{\text{S}}=1$. By replacing $p_{\text{S}}=q(\cc_{\text{S}})$ in \eqref{eq:psqs} and taking equality and $\cc_{\text{S}}=1$ we have
	\begin{align}\label{eq:q21}
		q(1)&=\frac{(1-\beta)(p+p_{\text{I}}+2)-p_{\text{I}}^2(1-\cc_{\text{I}}^2)-2p_{\text{I}}-1}{2a_{\text{S}}\left[p_{\text{I}}(1-\cc_{\text{I}})+\beta\right]}\notag\\&=\frac{\bar{p}^2-\left[p_{\text{I}}(1-\cc_{\text{I}})+\beta\right]\left[p_{\text{I}}(1+\cc_{\text{I}})+\beta\right]}{2a_{\text{S}}\left[p_{\text{I}}(1-\cc_{\text{I}})+\beta\right]} \; ,
	\end{align}
	where $\bar{p}=\frac{p2^{\bar{R}}}{2^{2\bar{R}}-1}$. Using this expression, we replace $\cc_{\text{S}}=1$ in the condition in \eqref{eq:kR} obtaining that $\cc_{\text{R}}=1$ if and only if $a_{\text{S}}\geq\xi$, with $\xi$ given by \eqref{eq:mu}. Finally, if $\cc_{\text{R}}<1$, we obtain, by taking equality in \eqref{eq:kR},
\begin{equation}
	q(\cc_{\text{R}})=\frac{p_{\text{I}}\cc_{\text{I}}}{p_{\text{I}}(\cc_{\text{R}}-\cc_{\text{I}})-a_{\text{S}}\cc_{\text{R}}\left(1-\frac{\beta}{a_{\text{S}}}\right)} \; .
\end{equation}
Replacing $p_{\text{S}}$ with the above $q(\cc_{\text{R}})$ in \eqref{eq:psqs}, and taking $\cc_{\text{S}}=\cc_{\text{R}}$ and equality, we obtain also a second-order equation for $\cc_{\text{R}}$ and thus a closed-form expression.

\bibliographystyle{IEEEtran}
\bibliography{SecondaryMAC}

\end{document}